\documentclass[journal, onecolumn]{IEEEtran}
\IEEEoverridecommandlockouts

\usepackage{bm}
\usepackage{amsmath,amssymb,amsfonts,graphicx,amsthm,mathtools,nicefrac,mathrsfs}
\usepackage{graphicx, subfigure}
\usepackage{times}
\usepackage{color}
\usepackage[table]{xcolor}
\usepackage{booktabs}
\usepackage{multirow}
\usepackage{lscape}

\DeclareMathOperator{\E}{\mathbb{E}}
\DeclareMathOperator{\dist}{dist}
\DeclareMathOperator{\sig}{sig}
\DeclareMathOperator{\cor}{cor}
\DeclareMathOperator{\rad}{rad}
\DeclareMathOperator{\diam}{diam}

\DeclareMathOperator{\cone}{cone}
\DeclareMathOperator{\diag}{diag}

\renewcommand{\Pr}[2][]{\mathbb{P}_{#1} \left\{ #2 \rule{0mm}{3mm}\right\}}
\newcommand{\ip}[2]{\left\langle#1,#2\right\rangle}
\newcommand{\gw}[1]{\omega\left(#1\right)}

\def \CC {\mathcal{C}}

\def \EE {\mathcal{E}}
\def \NN {\mathcal{N}}

\def \TT {\mathcal{T}}

\def \B {\mathbb{B}}
\def \R {\mathbb{R}}
\def \S {\mathbb{S}}

\def \va {\bm{a}}
\def \vb {\bm{b}}
\def \vd {\bm{d}}

\def \vg {\bm{g}}
\def \vh {\bm{h}}
\def \vs {\bm{s}}
\def \vu {\bm{u}}
\def \vv {\bm{v}}
\def \vw {\bm{w}}
\def \vx {\bm{x}}
\def \vy {\bm{y}}
\def \vz {\bm{z}}
\def \vzero {\bm{0}}

\def \mA {\bm{A}}
\def \mB {\bm{B}}
\def \mI {\bm{I}}
\def \mD {\bm{D}}
\def \mSigma {\bm{\Sigma}}
\def \mPhi {\bm{\Phi}}
\def \mPsi {\bm{\Psi}}

\newtheorem{theorem}{Theorem}
\newtheorem{lemma}{Lemma}

\newtheorem{corollary}{Corollary}
\newtheorem{condition}{Condition}

\newtheorem{fact}{Fact}

\theoremstyle{definition}

\theoremstyle{remark}
\newtheorem{remark}{Remark}

\makeatletter
\newtheorem*{rep@theorem}{\rep@title}
\newcommand{\newreptheorem}[2]{%
	\newenvironment{rep#1}[1]{%
		\def\rep@title{#2 \ref{##1}}%
		\begin{rep@theorem}}%
		{\end{rep@theorem}}}
\makeatother

\newreptheorem{theorem}{Theorem}
\newreptheorem{lemma}{Lemma}
\newreptheorem{corollary}{Corollary}
\newreptheorem{proposition}{Proposition}

\begin{document}
%
\title{Stable Recovery of Structured Signals From Corrupted Sub-Gaussian Measurements}
%
%
\author{
\IEEEauthorblockN{Jinchi Chen\IEEEauthorrefmark{1} and Yulong~Liu\IEEEauthorrefmark{1}}

\IEEEauthorblockA{\IEEEauthorrefmark{1}School of Physics, Beijing Institute of Technology, Beijing 100081, China}

\thanks{The material in this paper was presented in part at the IEEE International Symposium on Information Theory, Aachen, Germany, June 2017 \cite{chen2017Corrupted}.}
}

%

\maketitle


\begin{abstract}
This paper studies the problem of accurately recovering a structured signal from a small number of corrupted sub-Gaussian measurements. We consider three different procedures to reconstruct signal and corruption when different kinds of prior knowledge are available. In each case, we provide conditions (in terms of the number of measurements) for stable signal recovery from structured corruption with added unstructured noise. Our results theoretically demonstrate how to choose the regularization parameters in both partially and fully penalized recovery procedures and shed some light on the relationships among the three procedures. The key ingredient in our analysis is an extended matrix deviation inequality for isotropic sub-Gaussian matrices, which implies a tight lower bound for the restricted singular value of the extended sensing matrix. Numerical experiments are presented to verify our theoretical results.
\end{abstract}


\section{Introduction}
Corrupted sensing concerns the problem of recovering a structured signal from a relatively small number of corrupted measurements
\begin{align}\label{model: observe}
\vy = \bm{\Phi}\vx^{\star} +\vv^{\star} + \vz,
\end{align}
where $\bm{\Phi}\in\R^{m\times n}$ is the sensing matrix, $\vx^{\star}\in\R^n$ is the structured signal, $\vv^{\star}\in\R^m$ is the structured corruption, and $\vz \in \R^m$ is the unstructured observation noise. The goal is to estimate $\vx^{\star}$ and $\vv^{\star}$ from given knowledge of $\vy$ and $\bm{\Phi}$.

This problem has received increasing attention recently with many interesting practical applications as well as theoretical consideration. Examples of applications include face recognition \cite{wright2009robust}, subspace clustering \cite{elhamifar2009sparse}, sensor network \cite{haupt2008compressed}, latent variable modeling \cite{chandrasekaran2011rank}, and so on. On the theoretical side, performance guarantees include sparse signal recovery from sparse corruption \cite{wright2010dense,li2013compressed,nguyen2013exact,nguyen2013robust,kuppinger2012uncertainty,studer2012recovery,pope2013probabilistic,studer2014stable}, low-rank matrix recovery from sparse corruption \cite{chandrasekaran2011rank, candes2011robust, xu2012robust, xu2013outlier}, and structured signal recovery from structured corruption \cite{foygel2014corrupted}. It is worth noting that this model \eqref{model: observe} also includes the signal separation (or demixing) problem \cite{mccoy2014sharp} in which $\vv^{\star}$ might actually contain useful information and thus be necessary to be recovered. In particular, if there is no corruption $(\vv^{\star} = \vzero)$, this model \eqref{model: observe} reduces to the standard compressed sensing problem.

Since this problem is generally ill-posed, tractable recovery is possible when both signal and corruption are suitably structured. Typical examples of structured signals (or corruptions) include sparse
vectors, which exhibit low complexity under the $\ell_1$ norm, and low-rank matrices which exhibit low complexity with respect to the nuclear norm (more examples can be found in \cite{chandrasekaran2012convex}). Let $f(\cdot)$ and $g(\cdot)$ be suitable norms which promote structures for signal and corruption respectively. We consider three different convex optimization approaches to disentangle signal and corruption when different kinds of prior information are available. Specifically, when prior knowledge of either signal $f(\vx^{\star})$ or corruption $g(\vv^{\star})$ is available and the noise level $\delta$ (in terms of the $\ell_2$ norm) is known, it is natural to consider the following constrained convex recovery procedures
\begin{align}\label{Constrained_Optimization_I}
\min_{\vx, \vv} ~f(\vx),\quad\text{s.t.~}&g(\vv) \leq g(\vv^{\star}), ~~\|\vy-\bm{\Phi}\vx-\vv\|_2\leq \delta
\end{align}
and
\begin{align}\label{Constrained_Optimization_II}
\min_{\vx, \vv} ~g(\vv),\quad\text{s.t.~}&f(\vx) \leq f(\vx^{\star}), ~~\|\vy-\bm{\Phi}\vx-\vv\|_2\leq \delta.
\end{align}
When only the noise level $\delta$ is known, it is convenient to use the partially penalized convex recovery procedure
\begin{align}\label{Partially_Penalized_Optimization}
\min_{\vx,\vv}~f(\vx)+\lambda\cdot g(\vv),\quad\text{s.t.}\quad\|\vy-\bm{\Phi}\vx-\vv\|_2\leq \delta,
\end{align}
where $\lambda > 0$ is a regularization parameter. When there is no prior knowledge available, it is practical to utilize the fully penalized convex recovery procedure
\begin{align}\label{Fully Penalized Optimization}
\min_{\vx,\vv}\frac{1}{2}\|\vy-\bm{\Phi}\vx-\vv\|_2^2+\tau_1\cdot f(\vx)+\tau_2\cdot g(\vv),
\end{align}
where $\tau_1,\tau_2 > 0$ are some regularization parameters.

\subsection{Contribution}
This paper considers the problem of recovering structured signals from corrupted sub-Gaussian measurements. More precisely, we consider the following observation model:
\begin{itemize}
\item Sub-Gaussian measurements: we assume that each row $\bm{\Phi}_i$ of the sensing matrix $\bm{\Phi}$ is independent, centered, and sub-Gaussian random vector with
\begin{equation}
\label{model: subgaus property}
\|\bm{\Phi}_{i}\|_{\psi_2}\leq K/\sqrt{m}~~ \textrm{and} ~~ \E\bm{\Phi}_{i}^T\bm{\Phi}_{i} = \mI_n/m,
\end{equation}
where $\|\cdot\|_{\psi_2}$ denotes the sub-Gaussian norm (defined in \eqref{Sub-Gaussian_Definition}) and $\mI_n$ is the $n$-dimensional identity matrix;

\item Bounded or sub-Gaussian noise: the unstructured noise $\vz$ is assumed to be bounded $(\|\vz\|_2 \leq \delta)$ or be a random vector with independent centered sub-Gaussian entries satisfying
\begin{align} \label{model: noise}
\|\vz_i\|_{\psi_2} \leq L ~~\text{and}~~ \E\vz_i^2 = 1.
\end{align}
\end{itemize}

Under the above model assumptions, the contribution of this paper is twofold:
\begin{itemize}
\item First, we establish an extended matrix deviation inequality for isotropic sub-Gaussian matrices, which provides a powerful tool for analyzing the corrupted sensing problem;

\item Second, based on the deviation inequality, we present performance guarantees for all three convex recovery procedures. These results theoretically demonstrate how to select the regularization parameters in both partially and fully penalized recovery procedures and shed some light on the relationships among the three approaches.

\end{itemize}

\subsection{Related Work}
Within the past few years, there have been numerous studies to understand the theoretical aspect of the corrupted sensing problem. These works might be roughly put into the following three groups by the type of structures considered.

\subsubsection{Sparse Signal Recovery from Sparse Corruption} The specific problem of recovering a sparse vector from sparsely corrupted measurements has been analyzed under a variety of different model assumptions. For instance, Wright and Ma \cite{wright2010dense} study exact recovery by the partially penalized recovery procedure (with $f(\cdot)= g(\cdot) = \|\cdot\|_1, \lambda = 1$, and $\delta = 0$) under a bouquet model: $\vy=\mPhi \vx^{\star} + \vv^{\star}$, where the columns of $\mPhi$ are independent identically distributed (i.i.d.) samples from a Gaussian distribution $\NN(\bm{\mu}, {\nu^2}\mI_m/m)$ with $\|\bm{\mu}\|_2=1$ and $\|\bm{\mu}\|_{\infty}\leq {C}/{\sqrt{m}}$, and $\vv^{\star}$ has uniformly random signs. Li \cite{li2013compressed} considers stable recovery of $(\vx^{\star}, \vv^{\star})$ from $\vy=\mPhi\vx^{\star}+\vv^{\star}+\vz$ by the partially penalized recovery procedure (with $f(\cdot)= g(\cdot) = \|\cdot\|_1$ and $\lambda = 1/\sqrt{\log(n/m)+1}$), where $\mPhi$ is a Gaussian matrix with i.i.d. entries and $\vz$ is a bounded noise vector.  In the absence of noise, \cite{li2013compressed} also considers a class of sensing matrices, in which the rows of $\mPhi$ are sampled independently from a population obeying $\E\mPhi_i^{T}\mPhi_i = \mI_n$ and $\|\mPhi_i\|_{\infty} \leq c$. In contrast to \cite{li2013compressed}, Nguyen and Tran \cite{nguyen2013exact} also establish stable or exact recovery of $(\vx^{\star}, \vv^{\star})$ from $\vy=\mPhi\vx^{\star}+\vv^{\star}+\vz$ by the partially penalized recovery procedure, but with a different model assumption in which $\mPhi$ has columns sampled uniformly from an orthonormal matrix, $\vx^{\star}$ has uniformly random signs, $\vv^{\star}$ has uniformly distributed support, and $\vz$ is a bounded noise vector. In a subsequent paper \cite{nguyen2013robust}, Nguyen and Tran study stable recovery of $(\vx^{\star}, \vv^{\star})$ from $\vy=\mPhi\vx^{\star}+\vv^{\star}+\vz$ by the fully penalized recovery procedure (with $f(\cdot)= g(\cdot) = \|\cdot\|_1$), where $\mPhi$ is an Guassian sensing matrix with i.i.d. entries and $\vz$ is a Gaussian noise vector. Pope \textit{et al.} \cite{pope2013probabilistic} investigate the exact recovery from the observation model $\vy=\mPhi\vx^{\star}+\mPsi\vv^{\star}$, where $\mPhi$ and $\mPsi$ are deterministic matrices. The authors require uniform randomness in the support set of $\vx^{\star}$ or $\vv^{\star}$ and rely on certain incoherence properties of $\mPhi$ and $\mPsi$. These probabilistic recovery guarantees improve or refine the previous results in \cite{kuppinger2012uncertainty, studer2012recovery, studer2014stable}.

\subsubsection{Low-rank Matrix Recovery from Sparse Corruption}

Another line of work considers the problem of recovering a low-rank matrix from sparse corruption. Chandrasekaran \textit{et al.} \cite{chandrasekaran2011rank} study the exact decomposition of a given matrix into its sparse and low-rank components by the partially penalized recovery procedure (with $f(\cdot) = \|\cdot\|_*$, $g(\cdot) = \|\cdot\|_1$ and $\delta = 0$), where $\|\mD\|_*$ denotes the nuclear norm of $\mD$ and $\|\mD\|_1 = \sum_{i,j} |\mD_{i,j}|$ . The work of Cand\`{e}s \textit{et al.} \cite{candes2011robust} uses this model for robust principal component analysis and image processing applications. Modifications to the rank-sparsity model also find applications in robust statistics, see e.g., \cite{xu2012robust, xu2013outlier}.

\subsubsection{Structured Signal Recovery from Structured Corruption}

The most closely related to our current paper are the results by \cite{mccoy2014sharp, foygel2014corrupted, amelunxen2014living, Zhang2017}, in which the problem of recovering a structured signal from structured corruption is considered.
In \cite{mccoy2014sharp}, McCoy and Tropp analyze the constrained recovery procedures (\eqref{Constrained_Optimization_I} and \eqref{Constrained_Optimization_II}) in the setting where $\mPhi$ is a random orthogonal matrix ($m=n$) and measurements are noiseless ($\delta = 0$). In the same setting, the work of Amelunxen \textit{et al.} \cite{amelunxen2014living} establishes the sharpness of their recovery results, by specifying an appropriate threshold under which the constrained procedures fail with high probability. Foygel and Mackey \cite{foygel2014corrupted} study both the constrained recovery procedures (\eqref{Constrained_Optimization_I} and \eqref{Constrained_Optimization_II}) and the partially penalized recovery procedure \eqref{Partially_Penalized_Optimization} in the setting where $\mPhi$ is a Gaussian matrix ($m \leq n$ or $m > n$) with i.i.d. entries and  the noise is bounded. A recent work of Zhang \textit{et al.} \cite{Zhang2017} establishes the phase transition theory of the constrained recovery procedures under the setting of \cite{foygel2014corrupted}. In the present paper, we establish performance guarantees for all three convex recovery procedures in the setting where $\mPhi$ is a sub-Gaussian matrix and the noise is bounded or sub-Gaussian. These results solve a series of open problems in \cite{foygel2014corrupted}, for example, allowing non-Gaussian sensing matrix and stochastic unstructured noise in observation model \eqref{model: observe} and analyzing the fully penalized convex recovery procedure \eqref{Fully Penalized Optimization} for arbitrary structured signals and structured corruptions.

\subsection{Organization}

The remainder of the paper is organized as follows. We begin by reviewing some preliminaries that are necessary for our subsequent analysis in Section \ref{Preliminaries}. In Section \ref{Extended_Matrix_Deviation_Inequality}, we establish a mathematical tool which is critical for analyzing the corrupted sensing problem. Section \ref{PerformanceGuarantees} is devoting to presenting
 the main theoretical results of the paper.  We then present results from numerical simulations in Section \ref{Simulations}. We conclude the paper in Section \ref{Conclusion}.

\section{Preliminaries}\label{Preliminaries}
In this section, we review some useful concepts from convex geometry and high-dimensional probability that underlie our analysis. Throughout the paper, $\S^{n-1}$ and $\B_2^n$ denote the unit sphere and ball in $\R^n$ under the $\ell_2$ norm respectively, while $\B_f^n: = \{\vu \in \R^n : f(\vu) \leq 1\}$ is the unit ball in $\R^n$ under the norm defined by $f$. The compatibility constant between $f$ and the $\ell_2$ norm is defined as $\alpha_f := \sup_{\vu\neq 0}{f(\vu)}/{\|\vu\|_2}$. We use the notation $C, C', c, c', \textrm{etc.},$ to refer to positive constants, whose value may change from line to line.


\subsection{Convex Geometry}

\subsubsection{Subdifferential}
The \emph{subdifferential} of $f$ at $\vx$ is the set of vectors
\begin{equation*}\label{DefinitionofSubdifferential}
  \partial f(\vx) = \{ \vu \in \R^n: f(\vx + \vd) \geq f(\vx) + \langle \vu, \vd\rangle ~~ \textrm{for all} ~~\vd \in \R^n \}.
\end{equation*}
For any number $\kappa \geq 0$, we denote the scaled (by $\kappa$) subdifferential as $\kappa \cdot \partial f(\vx)  =  \{\kappa \vu : \vu \in \partial f(\vx)\}$.

\subsubsection{Tangent Cone and Normal Cone}
A cone is a set that is closed under multiplication by positive scalars. The \emph{tangent cone} of $f$ at $\vx$ is defined as the set of descent directions of $f$ at $\vx$
\begin{equation*}\label{DefinitionofTangentCone}
    \mathcal{T}_f = \{\vu \in \R^n: f(\vx+t \cdot \vu) \le f(\vx)~~\textrm{for some}~~t>0\}.
\end{equation*}
The \emph{normal cone} of $f$ at $\vx$ is the polar of the tangent cone, given by
\begin{equation*}\label{DefinitionofNormalCone}
    \mathcal{N}_f = \{\vu \in \R^n: \langle\vu, \vd\rangle \leq 0~~\textrm{for all}~~ \vd \in \mathcal{T}_f\},
\end{equation*}
which may be written as the cone hull of the subdifferential (if $0 \notin \partial f(\vx)$) \cite[Theorem 1.3.5]{hiriart1993convex}
\begin{align*}\label{subgradientNormalCone}
  \mathcal{N}_f  = \cone\{\partial f(\vx)\}= \{\vu \in \R^n: \vu \in t \cdot \partial f(\vx)~~\textrm{for some}~~t>0 \}.
\end{align*}
Fig. \ref{fig:Cone} illustrates related concepts in $\R^2$ for the $\ell_1$ norm.

\begin{figure*}
  	\centering
  		\includegraphics[width= .6\textwidth]{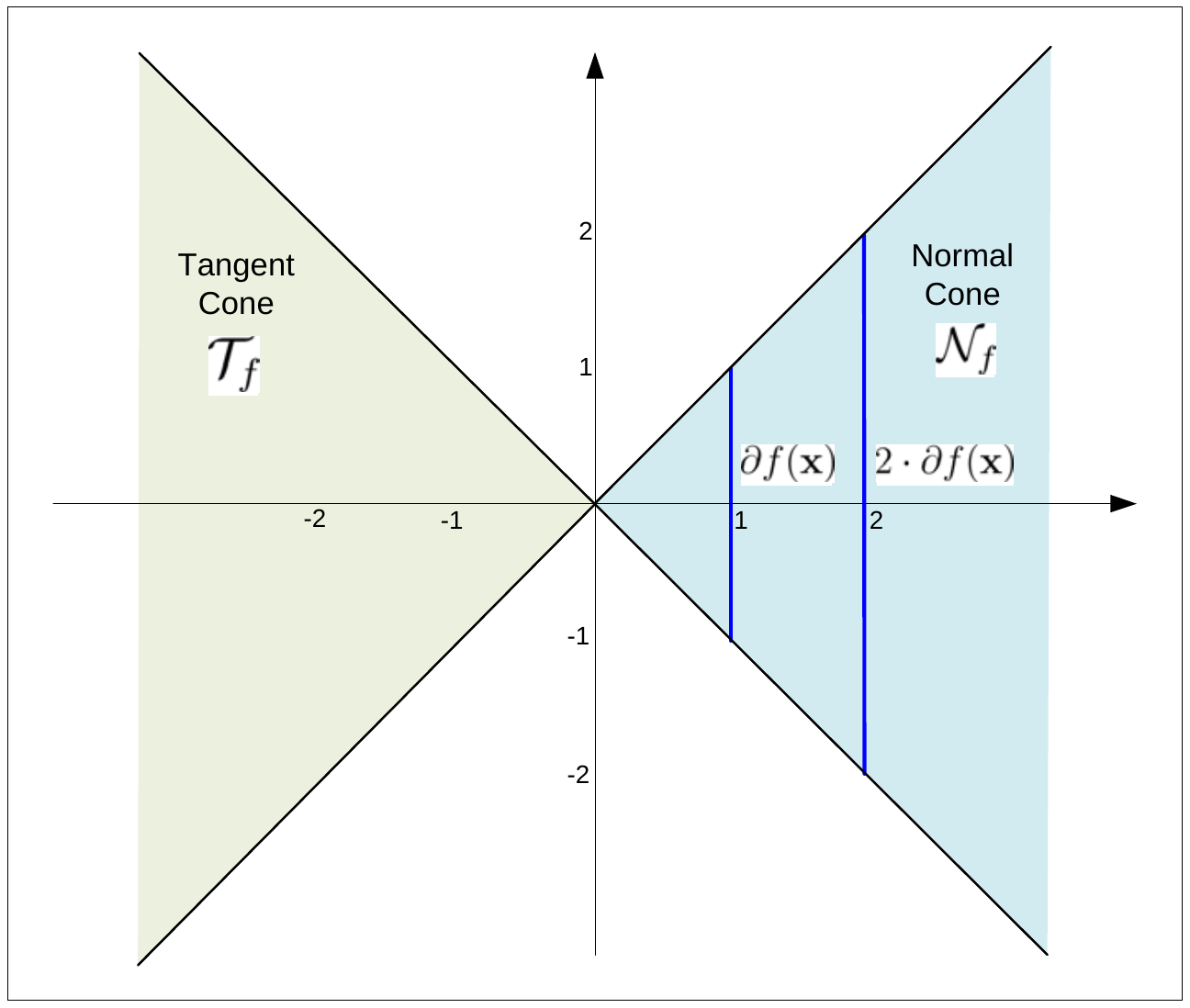}
	\caption{Illustrations of the subdifferential, scaled subdifferential, normal cone, and tangent cone for $f(\vx) = \|\vx\|_1$ at the point $\vx = (1, 0)$. (For clarity, cones are shifted to the original point.)}
	\label{fig:Cone}
\end{figure*}

\subsubsection{Gaussian Complexity, Gaussian Width, and Gaussian Distance}

For any $\TT\subseteq\R^{n}$, a simple way to quantify the ``size'' of $\TT$ is through the \emph{Gaussian complexity}
\begin{equation*}\label{Definition_Gaussian_complexity}
  \gamma(\TT) := \E \sup_{\vx \in \TT} |\langle \vg, \vx \rangle|, ~~ \textrm{where} ~~\vg\sim\NN(0,\mI_n).
\end{equation*}
Another popular geometric quantity closely related to Gaussian complexity is the \emph{Gaussian width}
\begin{equation*}\label{Definition_Gaussian_width}
  \omega(\TT) := \E \sup_{\vx \in \TT} \langle \vg, \vx \rangle, ~~ \textrm{where} ~~\vg\sim\NN(0,\mI_n).
\end{equation*}
In particular\footnote{Here we use a slightly sharper upper bound than that in \cite{liaw2016simple}. Indeed, since $\TT - \TT$ is original symmetric, then
$2\omega(\TT) = \gamma(\TT-\TT) = \E \sup_{\vx, \vy \in \TT} |\langle \vg, \vx-\vy \rangle| \geq \E \sup_{\vx \in \TT} |\langle \vg, \vx-\vy \rangle| \geq \E \sup_{\vx \in \TT} [|\langle \vg, \vx\rangle| - |\langle \vg, \vy\rangle|] \geq  \gamma(\TT) - (\E|\langle \vg, \vy\rangle|^2 )^\frac{1}{2}= \gamma(\TT) - \|\vy\|_2, ~ \forall~ \vy \in \TT.$ Rearranging yields the desired result.} \cite{liaw2016simple},
\begin{equation}\label{Relation}
  \left(\omega(\TT)+\|\vy\|_2\right)/3 \leq \gamma(\TT) \leq 2\omega(\TT)+\|\vy\|_2 ~~ \textrm{for every}~ \vy \in \TT.
\end{equation}
The \emph{Gaussian squared distance} $\eta^2(\TT)$ of a subset $\TT\subset\R^n$ is defined as
\begin{align*}
	\eta^2(\TT):=\E\inf_{\vu\in\TT}\|\vg-\vu\|_2^2,~~\textrm{where}~~\vg\sim\NN(0,\mI_n).
\end{align*}


\subsection{High-Dimensional Probability}


%
\subsubsection{Sub-Gaussian and Sub-exponential Random Variables}

A random variable $X$ is sub-Gaussian if its distribution is dominated by a Gaussian distribution. Formally, $X$ is called a \emph{sub-Gaussian} random variable if the Orlicz norm
\begin{equation}\label{Sub-Gaussian_Definition}
  \|X\|_{\psi_2} = \inf\{K > 0: \E \psi_2(|X|/K) \leq 1\}
\end{equation}
is finite for the Orlicz function $\psi_2(x) = \exp(x^2) -1$. The \emph{sub-Gaussian norm} of $X$, denoted $\|X\|_{\psi_2}$, is defined to be the smallest $K$ in \eqref{Sub-Gaussian_Definition}.
There are several equivalent definitions used in the literature. In particular, $X$ is sub-Gaussian if there exist absolute constants $C>1$ and $K$ such that
\begin{equation}\label{Sub-Gaussian_Definition1}
  \big[\E|X|^p\big]^{1/p} \leq K \sqrt{p},~~\text{ for all } p \geq 1
\end{equation}
or
\begin{equation}\label{Sub-Gaussian_Definition2}
    \Pr{\left| X \right| \geq t } \leq C \exp\left\{-{t^2}/K^2 \right\},~~\text{ for all } t \geq 0.
\end{equation}
Thus we can redefine the sub-Gaussian norm of $X$ as the smallest $K$ such that \eqref{Sub-Gaussian_Definition1} or \eqref{Sub-Gaussian_Definition2} holds. One can show that $\|X\|_{\psi_2}$ defined this way is within an absolute constant factor from that defined in \eqref{Sub-Gaussian_Definition}, see \cite[Section 5.2.3]{vershynin2010introduction}. Classical examples of sub-Gaussian random variables include Gaussian, Bernoulli, and all bounded random variables. Similarly, a random variable $X$ is sub-exponential if its distribution is dominated by an exponential distribution. More precisely, $X$ is called a \emph{sub-exponential} random variable if the Orlicz norm
\begin{equation}\label{Sub-exponential_Definition}
  \|X\|_{\psi_1} = \inf\{K > 0: \E \psi_1(|X|/K) \leq 1\}
\end{equation}
is finite for the Orlicz function $\psi_1(x) = \exp(x) -1$. The \emph{sub-exponential norm} of $X$, denoted $\|X\|_{\psi_1}$, is defined to be the smallest $K$ in \eqref{Sub-exponential_Definition}.

\subsubsection{Sub-Gaussian and Sub-exponential Random Vectors}

A random vector $\vx$ in $\R^n$ is sub-Gaussian (or sub-exponential) if all of its one-dimensional marginals are sub-Gaussian (or sub-exponential) random variables and its $\psi_2$-norm (or $\psi_1$-norm) is defined as
\begin{equation} \label{subGaussian vector}
\|\vx\|_{\psi_2}:=\sup_{\vy\in\S^{n-1}}\big\| \ip{\vx}{\vy} \big\|_{\psi_2} ~~~~\textrm{or} ~~~~ \|\vx\|_{\psi_1}:=\sup_{\vy\in\S^{n-1}}\big\| \ip{\vx}{\vy} \big\|_{\psi_1}.
\end{equation}
A random vector $\vx$ in $\R^n$ is called \emph{isotropic} if
\begin{equation*}\label{Isotropic_Definition}
\E(\vx\vx^T) = \mI_n.
\end{equation*}
For basic properties of sub-Gaussian and sub-exponential random variables (vectors), see e.g., \cite{vershynin2010introduction}, \cite[Chapter 2]{vershynin2016book}. Some properties which will be used in our proofs are included in Appendix A.

\subsubsection{Talagrand's Majorizing Measure Theorem}
The key ingredient in the proofs of our mathematical tools is the Talagrand's Majorizing Measure Theorem which states that any sub-Gaussian process is dominated by a Gaussian process with the same (or larger) increments.
\begin{fact}[Talagrand's Majorizing Measure Theorem]
	\label{Talagrand Them}
	Let $( X_{\vu} )_{\vu\in \TT}$ be a random process indexed by points in a bounded set $\TT \subset \R^{n}$. Assume that the process has sub-Gaussian increments, that is, there exists $M \geq 0$ such that
	\begin{equation*}
	\| X_{\vu} - X_{\vv} \|_{\psi_2} \leq M \|\vu-\vv\|_2 ~~~~ \text{for every} ~~ \vu,\vv \in \TT.
	\end{equation*}
    Then
	\begin{equation}\label{Expectation_Bound}
	\E \sup_{\vu,\vv \in \TT} \big| X_{\vu} - X_{\vv} \big| \leq C M \omega(\TT).
	\end{equation}
    Moreover, for any $t\geq 0$, the event
    \begin{equation}\label{High_Probability_Bound}
	\sup_{\vu, \vv \in \TT} \big|X_{\vu} - X_{\vv}\big| \leq CM \big[ \omega(\TT) + t \cdot \diam(\TT) \big]
	\end{equation}
    holds with probability at least $1- \exp(-t^2)$, where $\diam(\TT) := \sup_{\vx,\vy\in \TT}\|\vx-\vy\|_2$ denotes the diameter of $\TT$.
\end{fact}
\begin{remark}
	The expectation bound \eqref{Expectation_Bound} of this theorem can be found e.g. in \cite[Theorem~2.1.1, 2.1.5]{talagrand2006generic}. The high probability bound \eqref{High_Probability_Bound} can be found e.g. in \cite[Theorem~3.2]{dirksen2015tail} or \cite[Theorem~4.1]{liaw2016simple}.
\end{remark}

\section{Extended Matrix Deviation Inequality}\label{Extended_Matrix_Deviation_Inequality}
In this section, we establish a mathematical tool which provides a unified framework for analyzing all three recovery procedures. Note that if we denote $\bm{\Upsilon} = [\bm{\Phi}, \mI_m]$ and $\vs^{\star} = [(\vx^{\star})^T, (\vv^{\star})^T]^T$, then the observation model \eqref{model: observe} can be reformulated as $\vy = \bm{\Upsilon}\vs^{\star} + \vz $, which is the same as the standard compressed sensing model. It is now well known that the restricted singular value of the sensing matrix $\bm{\Upsilon}$ (see, e.g., \cite{bickel2009simultaneous, chandrasekaran2012convex}) plays a key role in analyzing the compressed sensing problem. With this observation in mind, we establish an extended matrix deviation inequality which implies a tight lower bound for the restricted singular value of the extended sensing matrix $[\bm{\Phi}, \mI_m]$. This result states that for any isotropic sub-Gaussian matrix $\mA$ and any bounded subset $\TT \subseteq \R^n\times\R^m$, the deviation of $\|\mA\va+\sqrt{m}\vb\|_2$ around $\sqrt{m}\cdot\sqrt{\|\va\|_2^2 + \|\vb\|_2^2}$ over vectors $(\va, \vb) \in \TT$ is uniformly bounded by the Gaussian complexity of $\TT$.

\begin{theorem}[Extended Matrix Deviation Inequality]
	\label{them: mat dev ineq}
	Let $\mA$ be an $m \times n$ matrix whose rows $\mA_i$ are independent centered isotropic sub-Gaussian vectors with $K = \max_i \|\mA_i\|_{\psi_2}$, and $\TT$ be a bounded subset of $\R^n\times\R^m$. Then
	\begin{align*}
	\E\sup_{(\va,\vb)\in \TT} \left| \|\mA\va+\sqrt{m}\vb\|_2 - \sqrt{m}\cdot\sqrt{\|\va\|_2^2 + \|\vb\|_2^2} \right| \leq CK^2\cdot\gamma(\TT).
	\end{align*}
	For any $t\geq 0$, the event
	\begin{align*}
	\sup_{ (\va,\vb)\in \TT }\left| \|\mA\va+\sqrt{m}\vb\|_2 - \sqrt{m}\cdot\sqrt{\|\va\|_2^2 + \|\vb\|_2^2} \right| \leq CK^2[ \gamma(\TT) + t\cdot \rad(\TT) ]
	\end{align*}
	holds with probability at least $1-\exp(-t^2)$, where $\rad(\TT) := \sup_{\vx \in \TT}\|\vx\|_2$ denotes the radius of $\TT$.
\end{theorem}

\begin{proof}
  See Appendix \ref{ProofofEMDI}.
\end{proof}

\begin{remark}[Anisotropic case] By using a simple linear transform, Theorem \ref{them: mat dev ineq} can be extended to the anisotropic case in which each row of $\mA$ satisfies $\E \mA_i^T \mA_i = \mSigma $ for some invertible covariance matrix $\mSigma$. To see this, consider the whitened version $\mB = \mA \mSigma^{-1/2}$.
Note that
\begin{align*}
\|\mB_i\|_{\psi_2} & = \sup_{\vy\in\S^{n-1}}\big\| \ip{\mB_i^T}{\vy} \big\|_{\psi_2}= \sup_{\vy\in\S^{n-1}}\left\| \ip{\mA_i^T}{\mSigma^{-1/2}\vy} \right\|_{\psi_2} \\
                   & = \sup_{\vy\in\S^{n-1}}\|\mSigma^{-1/2}\vy\|_2 \cdot \left\| \ip{\mA_i^T}{\mSigma^{-1/2}\vy/\|\mSigma^{-1/2}\vy\|_2} \right\|_{\psi_2} \\
                   &\leq \|\mSigma^{-1/2}\| \cdot K.
\end{align*}
Define the block diagonal matrix $\tilde{\mSigma} = \diag\{\mSigma^{1/2}, \mI_m\}$, then we have
\begin{align*}
  \E\sup_{(\va,\vb)\in \TT} \left| \|\mA\va+\sqrt{m}\vb\|_2 - \sqrt{m}\cdot\sqrt{\|\mSigma^{1/2}\va\|_2^2 + \|\vb\|_2^2} \right| &=  \E\sup_{(\va,\vb)\in \TT} \left| \|\mB\mSigma^{1/2}\va+\sqrt{m}\vb\|_2 - \sqrt{m}\cdot\sqrt{\|\mSigma^{1/2}\va\|_2^2 + \|\vb\|_2^2} \right|\\
   & = \E\sup_{(\va,\vb)\in \tilde{\mSigma}\TT} \left| \|\mB\va+\sqrt{m}\vb\|_2 - \sqrt{m}\cdot\sqrt{\|\va\|_2^2 + \|\vb\|_2^2} \right|\\
   & \leq C\|\mSigma^{-1}\| K^2 \gamma(\tilde{\mSigma}\TT)\\
   & \leq C\|\mSigma^{-1}\| \max\{\|\mSigma\|^{1/2}, 1\} K^2 \gamma(\TT).
\end{align*}
The first inequality follows from Theorem \ref{them: mat dev ineq}. The last inequality holds because $\gamma(\tilde{\mSigma}\TT) \leq \|\tilde{\mSigma}\| \gamma(\TT) = \max\{\|\mSigma\|^{1/2}, 1\} \gamma(\TT)$. Similarly, we can obtain the anisotropic version for the second part of Theorem \ref{them: mat dev ineq}.
\end{remark}

%

\begin{remark} Very recently, Liaw \textit{et al.} \cite{liaw2016simple} have shown that if $\mA$ is an isotropic (or anisotropic) sub-Gaussian matrix, then
\begin{equation*}
  \E \sup_{\va \in \TT} \big|\|\mA \va\|_2 - \sqrt{m}\|\va\|_2 \big| \leq C K^2 \gamma(\TT)
\end{equation*}
and the event
\begin{equation}\label{HighProbabilityBoundLiaw}
  \sup_{\va \in \TT} \big|\|\mA \va\|_2 - \sqrt{m}\|\va\|_2 \big| \leq C K^2 [\gamma(\TT)+t \rad(\TT)]
\end{equation}
 holds with probability at least $1-\exp(-t^2)$. It is not hard to see that these results can not be directly applied to the corrupted sensing problem, because the covariance matrix of each row of $[\bm{\Phi}, \mI_m]$ is singular. However, $[\bm{\Phi}, \mI_m]$ still has independent sub-Gaussian rows. This fact allows us to generalize the analysis in \cite{liaw2016simple} to the extended case. It is worth noting that there are several earlier variants of \eqref{HighProbabilityBoundLiaw} proved in \cite{klartag2005empirical,mendelson2007reconstruction,dirksen2015tail}. See \cite{liaw2016simple} for their comparisons.

\end{remark}

When $\TT$ is a subset of the unit sphere, we have the following corollary.

\begin{corollary}\label{CorollaryofMDI}
  Under the assumptions of Theorem \ref{them: mat dev ineq}, for any $t \geq 0$, the event
  \begin{equation}
  \sup_{ (\va,\vb)\in \TT \cap \S^{n+m-1}  }\left| \|\mA\va+\sqrt{m}\vb\|_2 - \sqrt{m} \right| \leq CK^2[ \gamma(\TT \cap \S^{n+m-1}) + t]
  \end{equation}
  holds with probability at least $1-\exp(-t^2)$.
\end{corollary}

Corollary \ref{CorollaryofMDI} may be specialized to bound to the restricted singular value of $[\mA, \sqrt{m}\mI_m]$. Indeed, let $\TT \subset \S^{n+m-1}$. Then the corollary states that, with high probability (e.g., $1-\exp \{-\gamma^2(\TT\cap \S^{n+m-1})\}$),
\begin{align}\label{LowerBound}
 \inf_{(\va, \vb) \in \TT \cap \S^{n+m-1}}\|\mA \va+ \sqrt{m}\vb\|_2  & \geq \sqrt{m} - CK^2{\gamma(\TT\cap \S^{n+m-1})}
\end{align}
and
\begin{align}\label{UpperBound}
 \sup_{(\va, \vb) \in \TT\cap \S^{n+m-1}}\|\mA \va+ \sqrt{m}\vb\|_2  & \leq \sqrt{m} + CK^2{\gamma(\TT\cap \S^{n+m-1})}
\end{align}
for all $(\va, \vb) \in \TT$. As we will see in the next section, \eqref{LowerBound} plays a key role in analyzing the corrupted sensing problem.

\section{Recovery From Corrupted Sub-Gaussian Measurements}\label{PerformanceGuarantees}
In this section, we present our theoretical results on the recovery of structured signals from corrupted sub-Gaussian measurements via three different convex recovery procedures.


\subsection{Recovery via Constrained Optimization}
We start with analyzing the constrained convex recovery procedures \eqref{Constrained_Optimization_I} and \eqref{Constrained_Optimization_II}.
Our first result shows that, with high probability, approximately
\begin{align}\label{NumberofMeasurements11}
CK^4 \left[ \omega^2(\TT_f(\vx^{\star})\cap\S^{n-1}) + \omega^2(\TT_g(\vv^{\star})\cap\S^{m-1}) \right]
\end{align}
corrupted measurements suffice to recover $(\vx^{\star}, \vv^{\star})$ exactly in the absence of noise and stably in the presence of noise, via either of the procedures \eqref{Constrained_Optimization_I} or \eqref{Constrained_Optimization_II}.

Before stating our result, we need to define the error set
\begin{equation*}
	\EE_1(\vx^{\star},\vv^{\star}):=\{(\va,\vb)\in \R^n\times\R^m: f(\vx^{\star}+\va)\leq f(\vx^{\star})  \text{ and } g(\vv^{\star}+\vb)\leq g(\vv^{\star}) \},
\end{equation*}
in which the error vector $(\hat{\vx}-\vx^{\star},\hat{\vv}-\vv^{\star})$ lives. By the convexity of $f$ and $g$, $\EE_1(\vx^{\star},\vv^{\star})$ belongs to the following convex cone
\begin{equation*}
\CC_1(\vx^{\star},\vv^{\star}):=\{(\va,\vb)\in \R^n\times\R^m: \langle \va, \vu\rangle \leq 0
\text{ and } \langle \vb, \vs\rangle \leq 0 \text{ for any }\vu\in\partial f(\vx^{\star})\text{ and }\vs\in \partial g(\vv^{\star}) \},
\end{equation*}
which is equivalent to
\begin{equation*}
      \{(\va,\vb)\in \R^n\times\R^m: \va \in \mathcal{T}_f(\vx^{\star})  \text{ and } \vb \in \mathcal{T}_g(\vv^{\star}) \}.
\end{equation*}
Fig. \ref{fig:Errorset} illustrates the error set $\EE_1$, error cone $\CC_1$, and corresponding spherical part of $\CC_1$. As we will see in this section, the Gaussian complexity of the spherical part of the error cone is closely related to the number of measurements to guarantee successful recovery.

\begin{figure*}
  	\centering
  		\includegraphics[width= .6\textwidth]{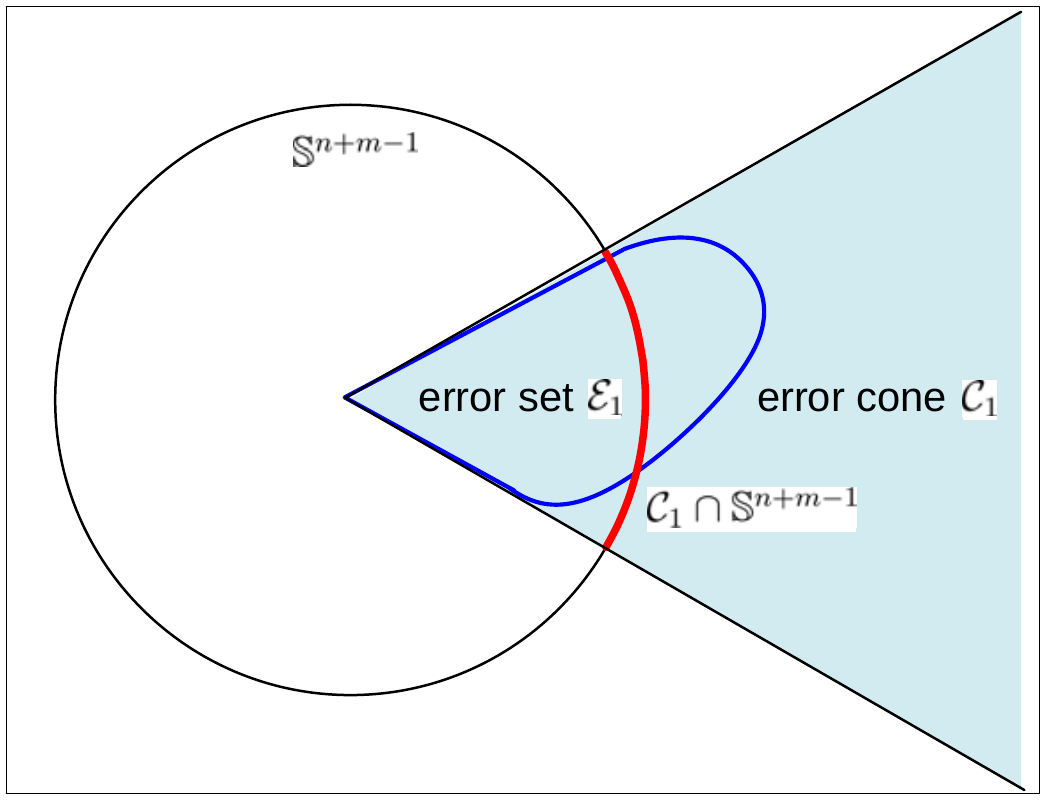}
	\caption{Illustrations of the error set $\EE_1$, error cone $\CC_1$, and corresponding spherical part of $\CC_1$: $\CC_1\cap\S^{n+m-1}$ (denoted by the red arc).}
	\label{fig:Errorset}
\end{figure*}

Then we have the following result.

\begin{theorem}[Constrained Recovery]
	\label{them: Constrained Recovery}
	Let $(\hat{\vx}, \hat{\vv})$ be the solution to either of the constrained optimization problems \eqref{Constrained_Optimization_I} or \eqref{Constrained_Optimization_II}. If the number of measurements
	\begin{align}\label{NumberofMeasurements1}
	\sqrt{m} \geq  CK^2 \gamma( \CC_1(\vx^{\star},\vv^{\star})\cap\S^{n+m-1} ) + \epsilon,
	\end{align}
	then
	\begin{align*}
	\sqrt{\|\hat{\vx}-\vx^{\star}\|_2^2+\|\hat{\vv}-\vv^{\star}\|_2^2}\leq \frac{2\delta\sqrt{m}}{\epsilon}
	\end{align*}
	with probability at least $1-\exp\{-\gamma^2(\CC_1\cap\S^{n+m-1})\}$.
\end{theorem}
\begin{remark}
  If we consider the observation model $\vy = \mA \vx^{\star} + \sqrt{m} \vv^{\star} + \vz$ as that in \cite{nguyen2013robust}, then the factor $\sqrt{m}$ can be removed from the error bound. This change of the model seems to make
  our result more interpretable, but we analyze the original observation model \eqref{model: observe} here.
  \end{remark}

\begin{proof}
    Since $(\hat{\vx}, \hat{\vv})$ solves \eqref{Constrained_Optimization_I} or \eqref{Constrained_Optimization_II}, we have $f(\hat{\vx}) \leq f(\vx^{\star})$ and $g(\hat{\vv}) \leq g(\vv^{\star})$. This implies $(\hat{\vx}-\vx^{\star},\hat{\vv}-\vv^{\star})\in\EE_1(\vx^{\star},\vv^{\star}) \subset  \CC_1(\vx^{\star},\vv^{\star})$.
    It then follows from \eqref{LowerBound} and \eqref{NumberofMeasurements1} that the event
	\begin{align}\label{LowerBound1}
	\min_{(\va,\vb)\in \CC_1(\vx^{\star},\vv^{\star})\cap\S^{n+m-1} }\sqrt{m}\|\bm{\Phi}\va+\vb\|_2  \geq \sqrt{m}-CK^2 \gamma( \CC_1(\vx^{\star},\vv^{\star})\cap\S^{n+m-1} ) \geq \epsilon
	\end{align}
    holds with probability at least $1-\exp\{-\gamma^2( \CC_1(\vx^{\star},\vv^{\star})\cap\S^{n+m-1} )\}$.

    On the other hand, since both $(\hat{\vx},\hat{\vv})$ and $(\vx^{\star},\vv^{\star})$ are feasible, by the triangle inequality, we have
	\begin{align} \label{prf them1 : step2}
	\|\bm{\Phi}(\hat{\vx}-\vx^{\star})+(\hat{\vv}-\vv^{\star})\|_2  \leq \|\vy-\bm{\Phi}\hat{\vx}-\hat{\vv}\|_2 + \|\vy-\bm{\Phi}\vx^{\star}-\vv^{\star}\|_2\leq 2\delta.
	\end{align}

    Combining \eqref{LowerBound1} and \eqref{prf them1 : step2} yields
    \begin{align*}
      2\delta & \geq \|\bm{\Phi}(\hat{\vx}-\vx^{\star})+(\hat{\vv}-\vv^{\star})\|_2 \\
              & = \sqrt{\|\hat{\vx}-\vx^{\star}\|_2^2+\|\hat{\vv}-\vv^{\star}\|_2^2} \cdot \left\|\frac{\bm{\Phi}(\hat{\vx}-\vx^{\star})}{\sqrt{\|\hat{\vx}-\vx^{\star}\|_2^2+\|\hat{\vv}-\vv^{\star}\|_2^2}}+\frac{(\hat{\vv}-\vv^{\star})}{\sqrt{\|\hat{\vx}-\vx^{\star}\|_2^2+\|\hat{\vv}-\vv^{\star}\|_2^2}}\right\|_2  \\
              & \geq \sqrt{\|\hat{\vx}-\vx^{\star}\|_2^2+\|\hat{\vv}-\vv^{\star}\|_2^2} \cdot \frac{\epsilon}{\sqrt{m}}.
    \end{align*}

    Rearranging completes the proof.

\end{proof}

To make use of Theorem \ref{them: Constrained Recovery}, it suffices to bound $\gamma(\CC_{1}(\vx^{\star}, \vv^{\star})\cap\S^{n+m-1})$. Since a number of upper bounds on Gaussian width and Gaussian squared distance for different structured settings are available in the literature (see, e.g., \cite{chandrasekaran2012convex,amelunxen2014living,foygel2014corrupted,thrampoulidis2015recovering}), it is highly desirable to bound $\gamma(\CC_{1}(\vx^{\star}, \vv^{\star})\cap\S^{n+m-1})$ in terms of these familiar parameters. Then we have the following result.


\begin{lemma}\label{BoundofGaussianComplexity1}
	The Gaussian complexity of $\CC_{1}(\vx^{\star}, \vv^{\star})\cap\S^{n+m-1}$ satisfies
	\begin{equation*}
	\gamma(\CC_{1}(\vx^{\star}, \vv^{\star})\cap\S^{n+m-1}) \leq  2\big[\gw{\TT_f(\vx^{\star})\cap\S^{n-1}} + \gw{\TT_g(\vx^{\star})\cap\S^{m-1}} + 1\big].
	\end{equation*}
\end{lemma}

\begin{proof}
Let $\vg\sim\NN(\vzero, \mI_n)$ and $\vh\sim\NN(\vzero, \mI_m)$, it follows from the definition of Gaussian complexity that
	\begin{align*}
	\gamma( \CC_{1}(\vx^{\star}, \vv^{\star})\cap\S^{n+m-1})
      &=\E\sup_{(\va,\vb)\in \CC_{1}(\vx^{\star}, \vv^{\star})\cap\S^{n+m-1} } \big|\ip{\vg}{\va} + \ip{\vh}{\vb} \big|\\
      &=\E\sup_{\substack{ c\in(0,1)\\\va\in\TT_f(\vx^{\star})\cap\S^{n-1}\\\vb\in\TT_g(\vv^{\star})\cap\S^{m-1}} }\left| \ip{\vg}{\va}\cdot c + \ip{\vh}{\vb}\cdot\sqrt{1-c^2} \right|\\
	&\leq \E\sup_{ \substack{ c\in(0,1)\\\va\in\TT_f(\vx^{\star})\cap\S^{n-1}\\\vb\in\TT_g(\vv^{\star})\cap\S^{m-1}} }c \cdot\left| \ip{\vg}{\va}\right| + \sqrt{1-c^2} \left|\ip{\vh}{\vb} \right|\\
	&\leq \E\sup_{\va\in\TT_f(\vx^{\star})\cap\S^{n-1}} \left| \ip{\vg}{\va} \right| + \E\sup_{\vb\in\TT_g(\vv^{\star})\cap\S^{m-1}} \left| \ip{\vh}{\vb} \right|\\
	&=\gamma(\TT_f(\vx^{\star})\cap\S^{n-1}) + \gamma( \TT_g(\vv^{\star})\cap\S^{m-1} )\\
    &\leq 2\big[\gw{\TT_f\cap\S^{n-1}} + \gw{\TT_g\cap\S^{m-1}} + 1\big].
	\end{align*}
The last inequality is due to \eqref{Relation}.
\end{proof}
Thus, \eqref{NumberofMeasurements11} follows from Theorem \ref{them: Constrained Recovery} and Lemma \ref{BoundofGaussianComplexity1}.

\begin{remark}
  When $\mPhi$ is a random orthogonal matrix ($m = n$) and measurements are noiseless ($\delta = 0$), Amelunxen \textit{et al.} \cite{amelunxen2014living} show that if $m \geq \iota(\TT_f(\vx^{\star})) + \iota(\TT_g(\vv^{\star}))$, then the constrained procedures \eqref{Constrained_Optimization_I}
 and \eqref{Constrained_Optimization_II} succeed with high probability, where the statistical dimension $\iota(\TT)$ of a closed convex cone $\TT \in \R^n$ is defined as $\iota(\TT) = \E \left( \sup_{\vu \in \TT \cap \B_2^n} \langle \vu, \vg \rangle   \right)^2 $ with $\vg \sim \NN(0, \mI_n)$. They also establish the sharpness of their result. Our result \eqref{NumberofMeasurements11} essentially coincides with this result because the statistical dimension of a closed convex cone $\TT \in \R^n$ is closely related to its Gaussian width \cite[Proposition 10.2]{amelunxen2014living}: $\omega^2(\TT\cap\S^{n-1}) \leq \iota(\TT) \leq \omega^2(\TT\cap\S^{n-1})+1$.

\end{remark}

\begin{remark}
 When $\mPhi$ is a Gaussian matrix ($m > n$ or $m \leq n$), Foygel and Mackey \cite{foygel2014corrupted} establish similar result as \eqref{NumberofMeasurements11} in terms of Gaussian squared complexity\footnote{The Gaussian squared complexity of a set $\TT \in \R^n$ is defined as $\zeta^2(\TT) = \E \left( \sup_{\vu \in \TT} \langle \vu, \vg \rangle \right)_{+}^2 $, where $\vg \sim \NN(0, \mI_n)$ and $(a)_+ = \max\{a,0\}$. Note that the Gaussian squared complexity $\zeta^2(\TT_f(\vx^{\star})\cap\S^{n-1})$ is very slightly bigger than the squared Gaussian width $\omega^2(\TT_f(\vx^{\star})\cap\S^{n-1})$.} under bounded noise ($\|\vz\|_2 \leq \delta$). They also numerically show that their result is sharp.
 In the absence of noise, Zhang \textit{et al.} \cite{Zhang2017} recently establish that the phase transition of the constrained procedure under Gaussian measurements occurs around $\omega^2(\TT_f(\vx^{\star})\cap\S^{n-1}) + \omega^2(\TT_g(\vv^{\star})\cap\S^{m-1})$.

\end{remark}

\subsection{Recovery via Partially Penalized Optimization}
We next present performance analysis for the partially penalized optimization problem \eqref{Partially_Penalized_Optimization}. Our second result shows that, with high probability, approximately
\begin{align}\label{NumberofMeasurements22}
CK^4 \left[ \eta^2(\lambda_1\cdot\partial f(\vx^{\star})) + \eta^2(\lambda_2\cdot\partial g(\vv^{\star})) \right]
\end{align}
corrupted measurements suffice to recover $(\vx^{\star}, \vv^{\star})$ exactly in the absence of noise and stably in the presence of noise, via the procedure \eqref{Partially_Penalized_Optimization}. Here, $\lambda_1$ and $\lambda_2$ are absolute constants such that the regularization parameter $\lambda = \lambda_2 / \lambda_1$.

In this case, it is natural to define the following error set
\begin{equation*}
\EE_2(\vx^{\star},\vv^{\star}):=\{(\va,\vb)\in\R^n\times\R^m: f(\vx^{\star}+\va)+\lambda\cdot g(\vv^{\star}+\vb)\leq f(\vx^{\star})+\lambda \cdot g(\vv^{\star}) \}.
\end{equation*}
By the convexity of $f$ and $g$, $\EE_2(\vx^{\star},\vv^{\star})$ belongs to the following convex cone
\begin{equation*}
\CC_2(\vx^{\star},\vv^{\star}):=\{(\va,\vb)\in \R^n\times\R^m:  \langle \va, \vu \rangle
+ \lambda \langle \vb, \vs \rangle \leq 0 \text{ for any }\vu\in \partial f(\vx^{\star})\text{ and } \vs\in \partial g(\vv^{\star}) \}.
\end{equation*}

Then we have the following result.

\begin{theorem}[Partially Penalized Recovery]
	\label{them: Partially_Penalized_Recovery}
	Let $(\hat{\vx}, \hat{\vv})$ be the solution to the partially penalized optimization problem \eqref{Partially_Penalized_Optimization}. If the number of measurements
	\begin{align}\label{NumberofMeasurements2}
	\sqrt{m} \geq  CK^2 \gamma( \CC_2(\vx^{\star},\vv^{\star})\cap\S^{n+m-1} ) + \epsilon,
	\end{align}
	then
	\begin{align*}
	\sqrt{\|\hat{\vx}-\vx^{\star}\|_2^2+\|\hat{\vv}-\vv^{\star}\|_2^2}\leq \frac{2\delta\sqrt{m}}{\epsilon}
	\end{align*}
	with probability at least $1-\exp\{-\gamma^2(\CC_2\cap\S^{n+m-1})\}$.
\end{theorem}

\begin{proof}
    The proof is similar to that of Theorem \ref{them: Constrained Recovery}.
    Since $(\hat{\vx}, \hat{\vv})$ solves \eqref{Partially_Penalized_Optimization}, we have $(\hat{\vx}-\vx^{\star}, \hat{\vv}-\vv^{\star}) \in \EE_2(\vx^{\star},\vv^{\star}) \subset \CC_2(\vx^{\star},\vv^{\star})$.
    It then follows from \eqref{LowerBound} and \eqref{NumberofMeasurements2} that the event
	\begin{equation}\label{LowerBound2}
	\min_{(\va,\vb)\in \CC_2(\vx^{\star},\vv^{\star})\cap\S^{n+m-1} }\sqrt{m}\|\bm{\Phi}\va + \vb\|_2 \geq \sqrt{m}-CK^2 \gamma( \CC_2(\vx^{\star},\vv^{\star})\cap\S^{n+m-1} ) \geq \epsilon
	\end{equation}
    holds with probability at least $1-\exp\{-\gamma^2( \CC_2(\vx^{\star},\vv^{\star})\cap\S^{n+m-1} )\}$.
	
    On the other hand, since both $(\hat{\vx},\hat{\vv})$ and $(\vx^{\star},\vv^{\star})$ are feasible, by the triangle inequality, we have
	\begin{equation} \label{prf them2 : step2}
	\|\bm{\Phi}(\hat{\vx}-\vx^{\star})+(\hat{\vv}-\vv^{\star})\|_2 \leq \|\vy-\bm{\Phi}\hat{\vx}-\hat{\vv}\|_2 + \|\vy-\bm{\Phi}\vx^{\star}-\vv^{\star}\|_2\leq 2\delta.
	\end{equation}
	
    Combining \eqref{LowerBound2} and \eqref{prf them2 : step2} completes the proof.
\end{proof}

To bound $\gamma(\CC_{2}(\vx^{\star}, \vv^{\star})\cap\S^{n+m-1})$ in terms of familiar parameters, we have the following result.
\begin{lemma}\label{BoundofGaussianComplexity2}
	The Gaussian complexity of $\CC_{2}(\vx^{\star}, \vv^{\star})\cap\S^{n+m-1}$ satisfies
	\begin{align*}
	\gamma(\CC_{2}(\vx^{\star}, \vv^{\star})\cap\S^{n+m-1}) \leq  2\sqrt{\eta^2(\lambda_1\cdot\partial f(\vx^{\star})) + \eta^2(\lambda_2\cdot\partial g(\vv^{\star}))} +1.
	\end{align*}
\end{lemma}

\begin{proof}
	By the definition of $\CC_2(\vx^{\star},\vv^{\star})$, for any point $(\va,\vb)\in\CC_2(\vx^{\star},\vv^{\star})$, we have
	\begin{align*}
	\langle \va, \vu \rangle + \lambda \langle \vb, \vs \rangle \leq 0
	\end{align*}
	for any $\vu\in\partial f(\vx^{\star})$ and $\vs\in\partial g(\vv^{\star})$.
	Multiplying both sides by $\lambda_1$ yields
	\begin{align*}
	\langle \va, \lambda_1\vu\rangle +  \langle \vb, \lambda_2\vs\rangle \leq 0.
	\end{align*}
    For any fixed $\vg \in \R^n$ and $\vh \in \R^m$, it follows from the Cauchy-Schwarz inequality that
	\begin{align*}
	\ip{\va}{\vg} + \ip{\vb}{\vh}  &\leq \ip{\va}{\vg-\lambda_1 \vu} + \ip{\vb}{\vh - \lambda_2 \vs }\\
	&\leq \|\va\|_2\|\vg-\lambda_1 \vu \|_2 + \|\vb\|_2\|\vh - \lambda_2 \vs\|_2.
	\end{align*}
	Choosing suitable $\bar{\vu} \in \partial f(\vx^{\star})$ and $\bar{\vs} \in \partial g(\vv^{\star})$ such that
	$$\|\vg-\lambda_1\cdot \bar{\vu}\|_2 = \dist(\vg,\lambda_1\cdot\partial f(\vx^{\star})):= d_f $$
	and
	$$\|\vh-\lambda_2\cdot \bar{\vs}\|_2 = \dist(\vh,\lambda_2\cdot\partial g(\vv^{\star})):= d_g,$$
    we obtain
    \begin{align}\label{GaussianWidthBound}
	\ip{\va}{\vg} + \ip{\vb}{\vh}
    & \leq \|\va\|_2\cdot \dist(\vg,\lambda_1\cdot\partial f(\vx^{\star})) + \|\vb\|_2\cdot \dist(\vh,\lambda_2\cdot\partial g(\vv^{\star}))\\ \notag
    & = d_f \cdot \|\va\|_2 + d_g \cdot \|\vb\|_2.
	\end{align}

    Then, by the definition of Gaussian width,
    \begin{align*}
	\gw{{\CC_2}(\vx^{\star},\vv^{\star})\cap\S^{n+m-1}}
    & =\E\sup_{(\va,\vb)\in \CC_2(\vx^{\star},\vv^{\star})\cap\S^{n+m-1}}\big[ \ip{\vg}{\va}+\ip{\vh}{\vb} \big] \\
	& \leq \E\sup_{(\va,\vb)\in \CC_2(\vx^{\star},\vv^{\star})\cap\S^{n+m-1}}\big[ \|\va\|_2\cdot d_f + \|\vb\|_2\cdot d_g \big]\\
	& \leq \E\sqrt{d_f^2 + d_g^2} \leq \sqrt{\E d_f^2 + \E d_g^2} =\sqrt{\eta^2(\lambda_1\cdot\partial f(\vx^{\star})) + \eta^2(\lambda_2\cdot\partial g(\vv^{\star}))},
	\end{align*}
where $\vg\sim\NN(\vzero, \mI_n)$ and $\vh\sim\NN(\vzero, \mI_m)$. The second and the third inequalities follow from the Cauchy-Schwarz and Jensen's inequalities, respectively. By \eqref{Relation}, we complete the proof.
\end{proof}

Thus, combining Theorem \ref{them: Partially_Penalized_Recovery} and Lemma \ref{BoundofGaussianComplexity2} yields \eqref{NumberofMeasurements22}.

\subsubsection{How to Choose the Regularization Parameter $\lambda$?}
The result \eqref{NumberofMeasurements22} also suggests a simple strategy to choose $\lambda_1$ and $\lambda_2$, and hence the regularization parameter $\lambda = \lambda_2/\lambda_1$. Observe that the number of observations to guarantee successful recovery by \eqref{Partially_Penalized_Optimization} depends on $\lambda_1$ and $\lambda_2$ through $\eta^2(\lambda_1\cdot\partial f(\vx^{\star}))$ and $\eta^2(\lambda_2\cdot\partial g(\vv^{\star}))$ respectively. In order to achieve the possibly smallest number of observations, it is natural to choose $\lambda_1$ and $\lambda_2$ which make the two Gaussian squared distances as small as possible. Thus we can pick
\begin{equation}\label{Chooselambda}
  \lambda_1^{\star} = \arg\min_{\lambda_1 > 0} \eta^2(\lambda_1\cdot\partial f(\vx^{\star})),~~ \lambda_2^{\star} = \arg\min_{\lambda_2 > 0} \eta^2(\lambda_2\cdot\partial g(\vv^{\star})),~~ \textrm{and}~~ \lambda^{\star} = \lambda_2^{\star}/\lambda_1^{\star}.
\end{equation}
Furthermore, it has been shown in \cite[Proposition 4.1]{amelunxen2014living} that if the subdifferential $\partial f(\vx^{\star})$ (or $\partial g(\vv^{\star})$) is nonempty, compact, and does not contain the original, then the function $J(\kappa) = \eta^2( \kappa \cdot\partial f(\vx^{\star})) $ (or $ \eta^2( \kappa \cdot\partial g(\vv^{\star})) $) is strictly convex for $\kappa > 0$, so it achieves its minimum at a unique point. Thus, under mild conditions, the optimal regularization parameter $\lambda^{\star}$ is always unique.

\begin{remark}
  In the setting of Gaussian measurements, Foygel and Mackey \cite{foygel2014corrupted} establish similar result as \eqref{NumberofMeasurements22} in terms of Gaussian squared distance under bounded noise ($\|\vz\|_2 \leq \delta$). Their result also indicates the same way to choose the regularization parameter $\lambda$. However, in the sub-Gaussian case, our proof is totally different from theirs. Because many useful properties of Gaussian distribution (e.g., the rotation invariance property, Gordon's comparison lemma \cite{gordon1988milman}) which are heavily used in \cite{foygel2014corrupted} do not hold in the sub-Gaussian case.
\end{remark}

\subsection{Recovery via Fully Penalized Optimization}
Finally, we analyze the fully penalized optimization problem \eqref{Fully Penalized Optimization}. In this case, due to the presence of noise, we require regularization parameters $\tau_1$ and $\tau_2$ to satisfy the following condition.
\begin{condition} \label{Assump: regular}
For any $\beta >1$,
\begin{align*}
\tau_1\geq \beta f^{\ast}(\bm{\Phi}^{T}\vz)\quad \text{and} \quad \tau_2\geq \beta g^{\ast}(\vz).
\end{align*}
\end{condition}
This condition is a natural extension to that in \cite[Lemma 1]{negahban2012}, where only the regularization parameter for signal ($\tau_1$) is considered and $\beta =2$.  Then our third result shows that, with high probability, approximately
\begin{align}\label{NumberofMeasurements33}
CK^4\left[\eta^2(\tau_1\cdot\partial f(\vx^{\star}))+\eta^2(\tau_2\cdot\partial g(\vv^{\star}))+\frac{(\tau_1\alpha_f)^2+(\tau_2\alpha_g)^2}{\beta^2}\right]
\end{align}
corrupted measurements suffice to recover $(\vx^{\star}, \vv^{\star})$ exactly in the absence of noise and stably in the presence of noise, via the procedure \eqref{Fully Penalized Optimization}.



Under the above condition, we similarly define the error set
   \begin{equation*}
    \EE_3(\vx^{\star},\vv^{\star}) :=  \{(\va,\vb)\in\R^n\times \R^m : \tau_1 f(\vx^{\star}+\va)+\tau_2 g(\vv^{\star}+\vb)  \leq \tau_1 f(\vx^{\star})+\tau_2  g(\vv^{\star}) + \frac{1}{\beta}[\tau_1 f(\va)+\tau_2 g(\vb)] \}.
   \end{equation*}
Note that if $ 0 < \beta \leq 1$, it then follows from the triangle inequality that the inequality defining $\EE_3(\vx^{\star},\vv^{\star})$ holds for all $(\va,\vb)$. Thus we require $\beta > 1$, which will restrict the set of $\EE_3(\vx^{\star},\vv^{\star})$ and yield the restricted error set. By the convexity of $f$ and $g$, $\EE_3(\vx^{\star},\vv^{\star})$ belongs to the following convex cone
\begin{equation*}
      \CC_3(\vx^{\star},\vv^{\star}):=\{(\va,\vb)\in \R^n\times\R^m: \tau_1 \langle \va, \vu\rangle + \tau_2 \langle \vb, \vs\rangle \leq \frac{1}{\beta}[\tau_1 f(\va) + \tau_2 g(\vb)] \text{ for any }\vu\in\partial f(\vx^{\star})\text{ and }\vs\in \partial g(\vv^{\star})\}.
\end{equation*}

Then we have the following result.
\begin{theorem}[Fully Penalized Recovery]
	\label{them: Fully Penalized Recovery}
	Let $(\hat{\vx}, \hat{\vv})$ be the solution to the fully penalized optimization problem \eqref{Fully Penalized Optimization} with $\tau_1$ and $\tau_2$ satisfying Condition \ref{Assump: regular}. If the number of measurements
	\begin{align}\label{NumberofMeasurements3}
	\sqrt{m} \geq  CK^2 \gamma( \CC_3(\vx^{\star},\vv^{\star})\cap\S^{n+m-1} ) + \epsilon,
	\end{align}
	then
	\begin{align}\label{ErrorBoundofFPR}
	\sqrt{\|\hat{\vx}-\vx^{\star}\|_2^2+\|\hat{\vv}-\vv^{\star}\|_2^2}\leq 2m \cdot \frac{\beta+1}{\beta} \cdot \frac{(\tau_1\alpha_f+\tau_2\alpha_g)}{\epsilon^2}
	\end{align}
	with probability at least $1-\exp\{-\gamma^2(\CC_3\cap\S^{n+m-1})\}$.
\end{theorem}
\begin{proof}
	Since $(\hat{\vx},\hat{\vv})$ solves \eqref{Fully Penalized Optimization}, we have
	\begin{equation}\label{OptimalCondition}
	\frac{1}{2}\|\vy-\bm{\Phi}\hat{\vx}-\hat{\vv}\|_2^2 + \tau_1 f(\hat{\vx})+\tau_2 g(\hat{\vv}) \leq \frac{1}{2}\|\vy-\bm{\Phi}\vx^{\star}-\vv^{\star}\|_2^2 +\tau_1 f(\vx^{\star})+\tau_2 g(\vv^{\star}).
	\end{equation}
    Observe that
	\begin{equation*}
	\frac{1}{2}\|\vy-\bm{\Phi}\hat{\vx}-\hat{\vv}\|_2^2  = \frac{1}{2}\|\bm{\Phi}(\hat{\vx}-\vx^{\star})+(\hat{\vv}-\vv^{\star})\|_2^2
    +\frac{1}{2}\|\vz\|_2^2  -\ip{\bm{\Phi}(\hat{\vx}-\vx^{\star})}{\vz} -\ip{\hat{\vv}-\vv^{\star}}{\vz}.
	\end{equation*}
    Substituting this into \eqref{OptimalCondition} yields
    \begin{equation}\label{InequalityofTheorem3}
     \frac{1}{2}\|\bm{\Phi}(\hat{\vx}-\vx^{\star})+(\hat{\vv}-\vv^{\star})\|_2^2  \leq \tau_1[f(\vx^{\star})-f(\hat{\vx})]
     +\tau_2 [g(\vv^{\star})-g(\hat{\vv})] + \ip{\bm{\Phi}(\hat{\vx}-\vx^{\star})}{\vz} + \ip{\hat{\vv}-\vv^{\star}}{\vz}.
    \end{equation}

    On the one hand, since $\|\bm{\Phi}(\hat{\vx}-\vx^{\star})+(\hat{\vv}-\vv^{\star})\|_2^2\geq 0$, we have
    \begin{align*}
    \tau_1 f(\hat{\vx})+\tau_2 g(\hat{\vv})
                 & \leq \tau_1 f(\vx^{\star}) + \tau_2 g(\vv^{\star}) + \ip{\bm{\Phi}(\hat{\vx}-\vx^{\star})}{\vz} + \ip{\hat{\vv}-\vv^{\star}}{\vz} \\
                 & \leq \tau_1 f(\vx^{\star}) + \tau_2 g(\vv^{\star}) + f^*(\bm{\Phi}^T\vz)\cdot f(\hat{\vx}-\vx^{\star})+ g^*(\vz)\cdot g(\hat{\vv}-\vv^{\star}) \\
                 & \leq \tau_1 f(\vx^{\star}) + \tau_2 g(\vv^{\star}) + \frac{\tau_1}{\beta}\cdot f(\hat{\vx}-\vx^{\star}) + \frac{\tau_2}{\beta}\cdot g(\hat{\vv}-\vv^{\star}),
    \end{align*}
    where $f^*(\cdot)$ and $g^*(\cdot)$ denote the dual norm\footnote{The dual norm of $f$ is defined as $f^{\ast}(\vd)=\sup_{\vu\in\B_f^n}\ip{\vu}{\vd}$, where $\B_f^n=\{ \vu\in\R^n : f(\vu)\leq 1\}$.}
 of $f(\cdot)$ and $g(\cdot)$ respectively. The second inequality follows from generalized H\"{o}lder's inequality. The last inequality is due to Condition \ref{Assump: regular}. This implies $(\hat{\vx}-\vx^{\star}, \hat{\vv}-\vv^{\star}) \in \EE_3(\vx^{\star},\vv^{\star}) \subset \CC_3(\vx^{\star},\vv^{\star})$.
    It then follows from \eqref{LowerBound} and \eqref{NumberofMeasurements3} that the event
	\begin{equation}\label{LowerBound3}
	\min_{(\va,\vb)\in \CC_3(\vx^{\star},\vv^{\star})\cap\S^{n+m-1} }\sqrt{m}\|\bm{\Phi}\va+\vb\|_2 \geq \sqrt{m}-CK^2 \gamma( \CC_3(\vx^{\star},\vv^{\star})\cap\S^{n+m-1} ) \geq \epsilon
	\end{equation}
    holds with probability at least $1-\exp\{-\gamma^2( \CC_3(\vx^{\star},\vv^{\star})\cap\S^{n+m-1} )\}$.

    On the other hand, it follows from \eqref{InequalityofTheorem3} that
    \begin{align}\label{UpperBound3}
    \frac{1}{2}\|\bm{\Phi}(\hat{\vx}-\vx^{\star})+(\hat{\vv}-\vv^{\star})\|_2^2
    & \leq \tau_1 \cdot f(\hat{\vx}-\vx^{\star})+\tau_2\cdot g(\hat{\vv}-\vv^{\star}) + \frac{\tau_1}{\beta}\cdot f(\hat{\vx}-\vx^{\star}) + \frac{\tau_2}{\beta}\cdot g(\hat{\vv}-\vv^{\star})\\ \notag
	&=\frac{\beta+1}{\beta}\big( \tau_1\cdot f(\hat{\vx}-\vx^{\star})+\tau_2\cdot g(\hat{\vv}-\vv^{\star}) \big)\\ \notag
	&= \frac{\beta+1}{\beta}\big( \alpha_f\tau_1\cdot \|\hat{\vx}-\vx^{\star}\|_2+\alpha_g\tau_2\cdot \|\hat{\vv}-\vv^{\star}\|_2 \big)\\ \notag
	& \leq \frac{\beta+1}{\beta}\cdot \sqrt{(\tau_1\alpha_f)^2+(\tau_2\alpha_g)^2} \cdot\sqrt{\|\hat{\vx}-\vx^{\star}\|_2^2+\|\hat{\vv}-\vv^{\star}\|_2^2} \\ \notag
    & \leq \frac{\beta+1}{\beta}\cdot (\tau_1\alpha_f+\tau_2\alpha_g) \cdot\sqrt{\|\hat{\vx}-\vx^{\star}\|_2^2+\|\hat{\vv}-\vv^{\star}\|_2^2},
	\end{align}
    where $\alpha_f$ and $\alpha_g$ are compatibility constants. The first inequality follows from the triangle inequality. In the last two inequalities, we have used the Cauchy-Schwarz inequality and the fact that $\sqrt{a^2+b^2} \leq a+b$ for $a, b \geq 0$, respectively.

    Combining \eqref{LowerBound3} and \eqref{UpperBound3} yields
    \begin{align*}
      \frac{2(\beta+1)}{\beta}\cdot (\tau_1\alpha_f+\tau_2\alpha_g) & \geq \sqrt{\|\hat{\vx}-\vx^{\star}\|_2^2+\|\hat{\vv}-\vv^{\star}\|_2^2} \cdot \left\|\frac{\bm{\Phi}(\hat{\vx}-\vx^{\star})}{\sqrt{\|\hat{\vx}-\vx^{\star}\|_2^2+\|\hat{\vv}-\vv^{\star}\|_2^2}}+\frac{(\hat{\vv}-\vv^{\star})}{\sqrt{\|\hat{\vx}-\vx^{\star}\|_2^2+\|\hat{\vv}-\vv^{\star}\|_2^2}}\right\|_2^2  \\
              & \geq \sqrt{\|\hat{\vx}-\vx^{\star}\|_2^2+\|\hat{\vv}-\vv^{\star}\|_2^2} \cdot \frac{\epsilon^2}{m}.
    \end{align*}

    Rearranging yields the desired result.

\end{proof}
	
To bound the Gaussian complexity of $\CC_{3}(\vx^{\star}, \vv^{\star})\cap\S^{n+m-1}$ in terms of familiar parameters, we have the following result.
\begin{lemma}\label{Gaussiancomplexity3}
	\begin{align*}
	\gamma(\CC_{3}(\vx^{\star}, \vv^{\star})\cap\S^{n+m-1}) \leq  2\left[\sqrt{\eta^2(\tau_1\cdot\partial f(\vx^{\star})) + \eta^2(\tau_2\cdot\partial g(\vv^{\star}))} + \frac{\sqrt{(\tau_1\alpha_f)^2+(\tau_2\alpha_g)^2}}{\beta}\right] + 1.
	\end{align*}
\end{lemma}

\begin{proof}
    The proof is similar to that of lemma \ref{BoundofGaussianComplexity2}. By the definition of $\CC_3(\vx^{\star},\vv^{\star})$, for any point $(\va,\vb)\in\CC_3(\vx^{\star},\vv^{\star})$, we have
	\begin{align*}
	\langle \va, \tau_1\vu \rangle + \langle \vb, \tau_2 \vs \rangle - \frac{1}{\beta}[\tau_1 f(\va) + \tau_2 g(\vb)] \leq 0
	\end{align*}
	for any $\vu\in\partial f(\vx^{\star})$ and $\vs\in\partial g(\vv^{\star})$.
    It then follows from the Cauchy-Schwarz inequality that
	\begin{align*}
	\ip{\va}{\vg} + \ip{\vb}{\vh}  &\leq \ip{\va}{\vg-\tau_1 \vu} + \ip{\vb}{\vh - \tau_2 \vs } + \frac{1}{\beta}[\tau_1 f(\va) + \tau_2 g(\vb)]\\
	&\leq \|\va\|_2\|\vg-\tau_1 \vu \|_2 + \|\vb\|_2\|\vh - \tau_2 \vs\|_2 + \frac{1}{\beta}[\tau_1 f(\va) + \tau_2 g(\vb)],
	\end{align*}
    where $\vg \in \R^n$ and $\vh \in \R^m$ are arbitrary vectors.
	Choosing suitable $\tilde{\vu} \in \partial f(\vx^{\star})$ and $\tilde{\vs} \in \partial g(\vv^{\star})$ such that
	$$\|\vg-\tau_1\cdot \tilde{\vu}\|_2 = \dist(\vg,\tau_1\cdot\partial f(\vx^{\star})):= d'_f $$
	and
	$$\|\vh-\tau_2\cdot \tilde{\vs}\|_2 = \dist(\vh,\tau_2\cdot\partial g(\vv^{\star})):= d'_g,$$
    we obtain
    \begin{align}\label{GaussianWidthBound}
	\ip{\va}{\vg} + \ip{\vb}{\vh}
    & \leq \|\va\|_2\cdot \dist(\vg,\tau_1\cdot\partial f(\vx^{\star})) + \|\vb\|_2\cdot \dist(\vh,\tau_2\cdot\partial g(\vv^{\star})) + \frac{1}{\beta}[\tau_1 f(\va) + \tau_2 g(\vb)]\\ \notag
    & \leq d'_f \cdot \|\va\|_2 + d'_g \cdot \|\vb\|_2 + \frac{1}{\beta}[\tau_1 \alpha_f \|\va\|_2 + \tau_2 \alpha_g \|\vb\|_2],
	\end{align}
    where the last inequality follows from the definition of the compatibility constant.

    Then, by the definition of Gaussian width,
	\begin{align*}
    \gw{ \CC_3(\vx^{\star},\vv^{\star})\cap\S^{n+m-1}} &=\E\sup_{(\va,\vb)\in\CC_3(\vx^{\star},\vv^{\star})\cap\S^{n+m-1}}\big[\ip{\vg}{\va}+\ip{\vh}{\vb} \big]\\
	& \leq \E\sup_{(\va,\vb)\in\CC_3(\vx^{\star},\vv^{\star})\cap\S^{n+m-1}} \big[d'_f \cdot \|\va\|_2 + d'_g \cdot \|\vb\|_2 + \frac{\tau_1 \alpha_f}{\beta} \cdot \|\va\|_2 + \frac{\tau_2 \alpha_g}{\beta}\cdot \|\vb\|_2 \big]\\
    & \leq \E \sqrt{d'^2_f + d'^2_g } + \frac{\sqrt{(\tau_1\alpha_f)^2+(\tau_2\alpha_g)^2}}{\beta} \\
	& \leq \sqrt{\eta^2(\tau_1\cdot\partial f(\vx^{\star})) + \eta^2(\tau_2\cdot\partial g(\vv^{\star}))} + \frac{\sqrt{(\tau_1\alpha_f)^2+(\tau_2\alpha_g)^2}}{\beta},
	\end{align*}
where $\vg\sim\NN(\vzero, \mI_n)$ and $\vh\sim\NN(\vzero, \mI_m)$. The second and the third inequalities follow from the Cauchy-Schwarz and Jensen's inequalities respectively. By \eqref{Relation}, we complete the proof.
\end{proof}

Clearly, \eqref{NumberofMeasurements33} comes from Theorem \ref{them: Fully Penalized Recovery} and Lemma \ref{Gaussiancomplexity3}.

\subsubsection{Identify the Range of $\tau_1$ and $\tau_2$}
Since our result relies on Condition \ref{Assump: regular}, it is necessary to identify the range of the regularization parameters $\tau_1$ and $\tau_2$ under different kinds of noise. In the absence of noise ($\vz = 0$), we can easily see that Condition \ref{Assump: regular} holds with $\tau_1 \geq 0$ and $\tau_2 \geq 0$. In general, we establish the following Chevet-type inequality which indicates Condition \ref{Assump: regular} holds with high probability under both bounded and stochastic noise scenarios.

\begin{lemma}
\label{lm: upper bound of general ip}
Let $\mA$ be an $m \times n$ matrix whose rows $\mA_i$ are independent centered isotropic sub-Gaussian vectors with $\max_{i} \|\mA_i\|_{\psi_2} \leq K$, and $\vw $ be any fixed vector. Let $\TT$ be any bounded subset $\R^n$. Then, for any $t\geq 0$, the event
\begin{align*}
	\sup_{\vu\in \TT} \ip{\mA\vu}{\vw} \leq CK \|\vw\|_2\big[ \gamma(\TT) + t\cdot\rad(\TT) \big]
\end{align*}
holds with probability at least $1-\exp\{-t^2\}$.
\end{lemma}

\begin{proof}
Define the random process
\begin{align*}
X_{\vu}:=\ip{\mA\vu}{\vw},\text{ for any }\vu \in \TT,
\end{align*}
which has sub-Gaussian increments
\begin{align*}
\|X_{\vu}-X_{\vu'}\|_{\psi_2} & = \|\ip{\mA(\vu-\vu')}{\vw}\|_{\psi_2} \\
                              &\leq  \|\vw\|_2 \|\mA(\vu-\vu')\|_{\psi_2}\\
                              &\leq CK \|\vw\|_2 \|\vu-\vu'\|_2
\end{align*}
for any $\vu,\vu' \in \TT$. The first inequality follows from the definition of the $\psi_2$-norm of a sub-Gaussian random vector \eqref{subGaussian vector} and the last inequality holds because $\{\ip{\mA_i^T}{(\vu-\vu')}\}_{i=1}^m$ are independent centered sub-Gaussian random variables with  $\|\ip{\mA_i^T}{(\vu-\vu')}\|_{\psi_2}\leq K \|\vu-\vu'\|_2$ and hence $\|\mA(\vu-\vu')\|_{\psi_2}\leq C K \|\vu-\vu'\|_2$.
Define $\bar{\TT} = \TT \cup \{\vzero\}$. It follows from Talagrand's Majorizing Measure Theorem that the event
\begin{align*}
\sup_{\vu \in \TT}\ip{\mA\vu}{\vw} &\leq \sup_{\vu \in \TT}|\ip{\mA\vu}{\vw}|= \sup_{\vu \in \bar{\TT}}|\ip{\mA\vu}{\vw}|\\
                                 & = \sup_{\vu \in \bar{\TT}}|\ip{\mA\vu}{\vw} - \ip{\mA\vzero}{\vw}|\\
                                 & \leq \sup_{\vu, \vu' \in \bar{\TT}}|\ip{\mA\vu}{\vw} - \ip{\mA\vu'}{\vw}|\\
                                 & \leq C'K\|\vw\|_2(\omega(\bar{\TT})+t\diam(\bar{\TT}))\\
                                 & \leq C''K \|\vw\|_2(\gamma(\TT)+t\rad(\TT))
\end{align*}
holds with probability at least $1-\exp\{-t^2\}$. In the last inequality, we have used the facts that $\omega(\bar{\TT}) \leq \gamma(\bar{\TT}) = \gamma(\TT)$ and $\diam(\TT) \leq 2\rad(\TT)$. This completes the proof.
\end{proof}

\textbf{Bounded Noise Case:} When the noise is bounded $(\|\vz\|_2\leq \delta)$, it follows from Lemma \ref{lm: upper bound of general ip} that the event (choosing $t = \sqrt{m}$)
\begin{align*}
	f^*(\bm{\Phi}^T\vz) = \sup_{\vu\in\B_f^n}\ip{\bm{\Phi}\vu}{\vz} \leq \frac{CK\delta}{\sqrt{m}}\big[ \gamma(\B_f^n)+ \sqrt{m} \cdot r_f \big]
\end{align*}
holds with probability at least $1-\exp(m)$, where $\B_f^n=\{\vu\in\R^n : f(\vu)\leq 1\}$ and $r_f=\sup\{\|\vu\|_2:\vu\in\B_f^n\}$. Thus it is natural to choose
\begin{equation}\label{tau1bounded}
  \tau_1 \geq \frac{CK\delta\beta}{\sqrt{m}}\big[ \gamma(\B_f^n)+\sqrt{m}\cdot r_f \big]:=\tau_{1B},
\end{equation}
 which implies that the first part of Condition \ref{Assump: regular} holds with high probability. In addition, note that $g^*(\vz)=\sup_{\vu\in\B_g^m}\ip{\vz}{\vu}\leq \delta \sup_{\vu\in\B_g^m}\|\vu\|_2=\delta\cdot r_g$, where $\B_g^m=\{\vu\in\R^m : g(\vu)\leq 1\}$ and $r_g=\sup\{\|\vu\|_2:\vu\in\B_g^m\}$. Therefore, we can choose
 \begin{equation}\label{tau2bounded}
   \tau_2 \geq \beta \delta\cdot r_g := \tau_{2B}.
 \end{equation}

\textbf{Sub-Gaussian Noise Case:} When $\vz$ is a sub-Gaussian random vector satisfying \eqref{model: noise}, then $\|\vz\|_2$ concentrates near the value $\sqrt{m}$ with high probability \cite[Theorem 3.1.1]{vershynin2016book}, namely $\|\|\vz\|_2 - \sqrt{m}\|_{\psi_2} \leq CK^2$. This implies
\begin{equation*}
  \Pr{\|\vz\|_2 \geq (L^2+1)\sqrt{m}} \leq \Pr{\big|\|\vz\|_2 - \sqrt{m}\big| \geq L^2\sqrt{m} } \leq 2 e^{-cm}.
\end{equation*}
Combining this with Lemma \ref{lm: upper bound of general ip} and taking union bound yields
\begin{align*}
	f^*(\bm{\Phi}^T\vz) = \sup_{\vu\in\B_f^n}\ip{\bm{\Phi}\vu}{\vz} \leq {CK(1+L^2)}\big[ \gamma(\B_f^n)+ \sqrt{m} \cdot r_f \big]
\end{align*}
with probability at least $1-3e^{-cm}$. Thus we can choose
\begin{equation}\label{tau1stochastic}
  \tau_1 \geq {CK(1+L^2)\beta}\big[ \gamma(\B_f^n)+ \sqrt{m} \cdot r_f \big]:=\tau_{1S},
\end{equation}
which means that the first part of Condition \ref{Assump: regular} holds with high probability in the sub-Gaussian noise case. For the second part of Condition \ref{Assump: regular}, note that the random process
\begin{align*}
X_{\va}:=\ip{\va}{\vz}, \quad\text{ for any }\va \in \TT
\end{align*}
has sub-Gaussian increments
\begin{align*}
\|X_{\va}-X_{\va'}\|_{\psi_2} &= \|\ip{\va-\va'}{\vz}\|_{\psi_2}\leq  L \cdot\|\va-\va'\|_2 ~~\textrm{ for any}~~ \va,\va' \in \TT.
\end{align*}
It follows from Talagrand's Majorizing Measure Theorem that the event (similar arguments to that of Lemma \ref{lm: upper bound of general ip})
\begin{align*}
\sup_{\va \in \TT}\ip{\va}{\vz} \leq CL [\gamma(\TT)+t\rad(\TT)]
\end{align*}
holds with probability at least $1-\exp(-t^2)$. Choosing $t = \sqrt{m}$, we obtain that the event
\begin{align*}
  g^*(\vz)=\sup_{\vu\in\B_g^m}\ip{\vz}{\vu}\leq {CL}\big[ \gamma(\B_g^m)+ \sqrt{m} \cdot r_g \big]
\end{align*}
holds with probability at least $1-\exp\{-m\}$. Thus, in the sub-Gaussian noise case, we can choose
\begin{equation}\label{tau2stochastic}
  \tau_2 \geq {CL\beta}\big[ \gamma(\B_g^m)+ \sqrt{m} \cdot r_g \big] :=\tau_{2S}.
\end{equation}

\subsubsection{How to Choose the Regularization Parameters $\tau_1$ and $\tau_2$?}\label{Howtochoosetau}
Once the range of $\tau_1$ and $\tau_2$ is determined, then we can choose the ``best'' regularization parameters under certain criteria in this range. Our theoretical results show that both the number of observations \eqref{NumberofMeasurements33} and the recovery error bound \eqref{ErrorBoundofFPR} rely on the regularization parameters $\tau_1$ and $\tau_2$. Then one might expect to select $\tau_1$ and $\tau_2$ in the specified range such that these two quantities are as small as possible. However, it is not hard to see that these two quantities do not achieve their minima at the same time in general, so it is practical to choose some criteria which make a tradeoff between the number of observations and the recovery error bound. Specifically, if one wants to choose $\tau_1$ and $\tau_2$ such that the number of observations \eqref{NumberofMeasurements33} which guarantees successful recovery for \eqref{Fully Penalized Optimization} is as small as possible, then one can pick
\begin{equation}\label{ChoosetauCase1}
  \tau_{1\cdot}^{\star} = \arg\min_{\tau_1 \geq \tau_{1\cdot}} \left\{ \eta^2(\tau_1\cdot\partial f(\vx^{\star})) + \frac{\alpha_f^2}{\beta^2}\tau_1^2 \right\} ~~ \textrm{and} ~~ \tau_{2\cdot}^{\star} = \arg\min_{\tau_2 \geq \tau_{2\cdot}} \left\{ \eta^2(\tau_2\cdot\partial g(\vv^{\star})) + \frac{\alpha_g^2}{\beta^2}\tau_2^2\right\},
\end{equation}
where $\tau_{1\cdot}$ (or $\tau_{2\cdot}$) denotes $\tau_{1B}$ (or $\tau_{2B}$) in the bounded noise case and $\tau_{1S}$ (or $\tau_{2S}$) in the sub-Gaussian noise case. Moreover, when both $\partial f(\vx^{\star})$ and $\partial g(\vv^{\star})$ are nonempty, compact, and do not contain the original, by the strict convexity of objective functions,  these ``best'' parameters are uniquely attained. If one wants to select $\tau_1$ and $\tau_2$ which make the recovery error bound \eqref{ErrorBoundofFPR} as small as possible, then one can pick
\begin{equation}\label{ChoosetauCase2}
  \tau_{1\cdot}^{\star} = \arg\min_{\tau_1 \geq \tau_{1\cdot}} \frac{2m(\beta+1)\alpha_f}{\beta\epsilon^2} \tau_1 =  \tau_{1\cdot}~~ \textrm{and} ~~ \tau_{2\cdot}^{\star} = \arg\min_{\tau_2 \geq \tau_{2\cdot}} \frac{2m(\beta+1)\alpha_g}{\beta\epsilon^2} \tau_2 =  \tau_{2\cdot}.
\end{equation}
This result demonstrates that if we want to recover signal and corruption accurately by the fully penalized procedure in the noise-free case, we require
\begin{equation}\label{taunoisefreecase}
  \tau_1 \rightarrow 0_{+}  ~~ \textrm{and} ~~ \tau_2 \rightarrow 0_{+},
\end{equation}
which might be regarded as promoting the equality constraint $\vy = \mPhi \vx + \vv $ in the fully penalized optimization problem.

\subsection{Relationships Among Three Procedures}\label{Relationships}
The theory of Lagrange multipliers \cite[Section 28]{rockafellar2015convex} asserts that the three procedures are essentially equivalent when the regularization parameters $\lambda, \tau_1,$ and $\tau_2$ are chosen correctly. However, the primary difficulty lies in determining these parameters when prior information is unavailable. Our results, \eqref{NumberofMeasurements22}, \eqref{NumberofMeasurements33}, and \eqref{ErrorBoundofFPR}, suggest some useful strategies to choose these parameters and shed some light on the relationships among these approaches.

Let $m_1, m_2,$ and $m_3$ be the necessary numbers of observations which guarantee successful recovery for constrained, partially penalized, and fully penalized recovery procedures, respectively. Since $\CC_{1}(\vx^{\star}, \vv^{\star}) \subseteq  \CC_{2}(\vx^{\star}, \vv^{\star}) \subseteq \CC_{3}(\vx^{\star}, \vv^{\star})$, by the definition of Gaussian complexity, we have $\gamma(\CC_{1}(\vx^{\star}, \vv^{\star})\cap\S^{n+m-1}) \leq \gamma(\CC_{2}(\vx^{\star}, \vv^{\star})\cap\S^{n+m-1}) \leq \gamma(\CC_{3}(\vx^{\star}, \vv^{\star})\cap\S^{n+m-1})$ and hence $m_1 \leq m_2 \leq m_3$. This relation seems natural, because one might expect that the more prior information, the less number of measurements to guarantee success.

On the other hand, it follows from \cite[Theorem 4.3, Proposition 10.2]{amelunxen2014living} and \cite[Proposition 1]{foygel2014corrupted} that, under a weak decomposition assumption\footnote{For $\vx \neq \vzero$, $\partial f(\vx)$ satisfies the weak decomposition assumption: $\exists \vw_0 \in  \partial f(\vx) ~\textrm{s.t.} ~\langle \vw-\vw_0, \vw_0 \rangle =0, \forall \vw \in \partial f(\vx)$.}, the optimized Gaussian squared distance $\min_{\kappa \geq 0} \eta^2( \kappa \cdot\partial f(\vx^{\star}))$ provides a faithful approximation to the statistical dimension $\iota(\TT_f(\vx^{\star}))$, and hence the squared Gaussian width $\omega^2(\TT_f(\vx^{\star})\cap\S^{n-1})$, i.e., $\min_{\kappa \geq 0} \eta^2( \kappa \cdot\partial f(\vx^{\star})) \approx  \omega^2(\TT_f(\vx^{\star})\cap\S^{n-1})$. This implies that
\begin{equation*}
  \eta^2( \lambda_1^{\star} \cdot\partial f(\vx^{\star})) \approx  \omega^2(\TT_f(\vx^{\star})\cap\S^{n-1}) ~~ \textrm{and} ~~ \eta^2( \lambda_2^{\star} \cdot\partial g(\vv^{\star})) \approx  \omega^2(\TT_g(\vv^{\star})\cap\S^{m-1}),
\end{equation*}
provided that both $\partial f(\vx^{\star})$  and $\partial g(\vv^{\star})$ are nonempty, compact, and do not contain the original.
It then follows from \eqref{NumberofMeasurements11} and \eqref{NumberofMeasurements22} that $m_1 \approx m_2$ if we choose $\lambda$ according to \eqref{Chooselambda}. For the fully penalized procedure, if we pick $\tau_1$ and $\tau_2$ according to \eqref{ChoosetauCase1} under the noise-free case, i.e.,
\begin{equation*}
  \tau_{1}^{\star} = \arg\min_{\tau_1 > 0} \left\{\eta^2(\tau_1\cdot\partial f(\vx^{\star})) + \frac{\alpha_f^2}{\beta^2}\tau_1^2\right\} ~~ \textrm{and} ~~ \tau_{2}^{\star} = \arg\min_{\tau_2 > 0} \left\{\eta^2(\tau_2\cdot\partial g(\vv^{\star})) + \frac{\alpha_g^2}{\beta^2}\tau_2^2\right\},
\end{equation*}
and if $\beta$ is chosen large enough such that both $\alpha_f^2\tau_1^2/\beta^2$ and $\alpha_g^2\tau_2^2/\beta^2$ are negligible, then we have $\tau_{1}^{\star} \approx \lambda_1^{\star}$ and $\tau_{2}^{\star} \approx \lambda_2^{\star}$, and hence $m_2 \approx m_3$. Thus we might conclude that under proper parameter selection strategies, the three approaches are approximately equivalent in terms of the necessary number of observations to guarantee success.

\section{Numerical Simulations}\label{Simulations}
In this section, we provide a series of numerical tests to verify our theoretical results in Section \ref{PerformanceGuarantees}. We present constrained, partially penalized, and fully penalized recovery experiments for sparse signals and sparse corruptions in the absence or presence of noise. In each experiment, we employ the CVX Matlab package \cite{grant2008cvx, grant2008graph} to solve our convex optimization problems.

\subsection{Phase Transition of Constrained Recovery Procedure}
We first investigate the empirical behavior of the constrained recovery procedure when the noise level $\delta = 0$ and the norm of the true signal $f(\vx^{\star}) = \|\vx^{\star}\|_1$ are known exactly. We fix the sample size and signal length $m = n = 128$, and vary the sparsity levels $(s_{\sig}, s_{\cor}) \in [1, 128] \times [1,128]$. We implement the following experiment $20$ times for each $(s_{\sig}, s_{\cor})$ pair:
\begin{itemize}
	\item [(1)] Generate a signal vector $\vx^{\star}$ with $s_{\sig}$ independent standard normal entries and set the other $n-s_{\sig}$ entries to 0.
	\item [(2)] Generate a corruption vector $\vv^{\star}$ with $s_{\cor}$ independent standard normal entries and set the other $m-s_{\cor}$ entries to 0.
	\item [(3)] For Gaussian measurements, draw a sensing matrix $\mPhi\in\R^{m\times n}$ with independent $\NN(0, 1/m)$ entries; for scaled symmetric Bernoulli measurements, draw a sensing matrix $\mPhi\in\R^{m\times n}$ with independent entries obeying $\Pr {\mPhi_{i,j} = 1/\sqrt{m}} = 1/2$ and $\Pr{\mPhi_{i,j} = -1/\sqrt{m}} = 1/2$.
	\item [(4)] Solve the following constrained optimization problem:
                \begin{equation}\label{L1constrained}
                  (\hat{\vx}, \hat{\vv}) \in \arg \min_{\vx,\vv} \|\vv\|_1,~ \textrm{s.t.}~~ \|\vx\|_1 \leq \|\vx^{\star}\|_1,~ \vy = \mPhi \vx + \vv.
                 \end{equation}
    \item [(5)] Declare success if $\|\hat{\vx}-\vx^{\star}\|_2 / \|\vx^{\star}\|_2 \leq 10^{-3}$.
\end{itemize}

To compare these empirical results with our theory, we overlay the theoretical recovery threshold suggested by \eqref{NumberofMeasurements11}\footnote{In \cite{Zhang2017}, the second author of this paper and his collaborators have shown that the phase transition of the constrained procedure under Gaussian measurements occurs around $\omega^2(\TT_f(\vx^{\star})\cap\S^{n-1}) + \omega^2(\TT_g(\vv^{\star})\cap\S^{m-1})$. A large number of numerical experiments indicate that this theoretical recovery threshold also holds for sub-Gaussian measurements. This observed universality phenomenon suggests us to set $CK^4 = 1$ in \eqref{NumberofMeasurements11} for our numerical simulations.}
\begin{equation}\label{threshold}
 \omega^2(\TT_f(\vx^{\star})\cap\S^{n-1}) + \omega^2(\TT_g(\vv^{\star})\cap\S^{m-1}),
\end{equation}
where the squared Gaussian widths of signals and corruptions can be accurately estimated  by (see, e.g., \cite[Appendix C]{chandrasekaran2012convex} or \cite[eq.(10)]{foygel2014corrupted})
\begin{align*}
&\min_{t\geq 0} \E \left[ \inf_{\vu\in t \cdot \partial \|\vx^{\star}\|_1 } \|\vg-\vu\|_2^2 \right]=\min_{t\geq 0} s_{\sig}(1+t^2) + \frac{2(n-s_{\sig})}{\sqrt{2\pi}}\left( (1+t^2)\int_{t}^{\infty} e^{-x^2/2}dx -te^{-t^2/2} \right)
\end{align*}
and
\begin{align*}
&\min_{t\geq 0} \E \left[ \inf_{\vu\in t\cdot\partial \|\vv^{\star}\|_1 } \|\vg-\vu\|_2^2 \right]=\min_{t\geq 0} s_{\cor}(1+t^2) + \frac{2(m-s_{\cor})}{\sqrt{2\pi}}\left( (1+t^2)\int_{t}^{\infty} e^{-x^2/2}dx -te^{-t^2/2} \right),
\end{align*}
respectively.

Fig. \ref{fig:Constrained_Procedure} displays the empirical probability of success for each setting of $(s_{\sig}, s_{\cor})$ averaged over the $20$ runs. We can see that our theoretical recovery threshold closely aligns with observed phase transition under both Gaussian and Bernoulli measurements.

\begin{figure*}
  	\centering
	\subfigure{
  		\includegraphics[width= .48\textwidth]{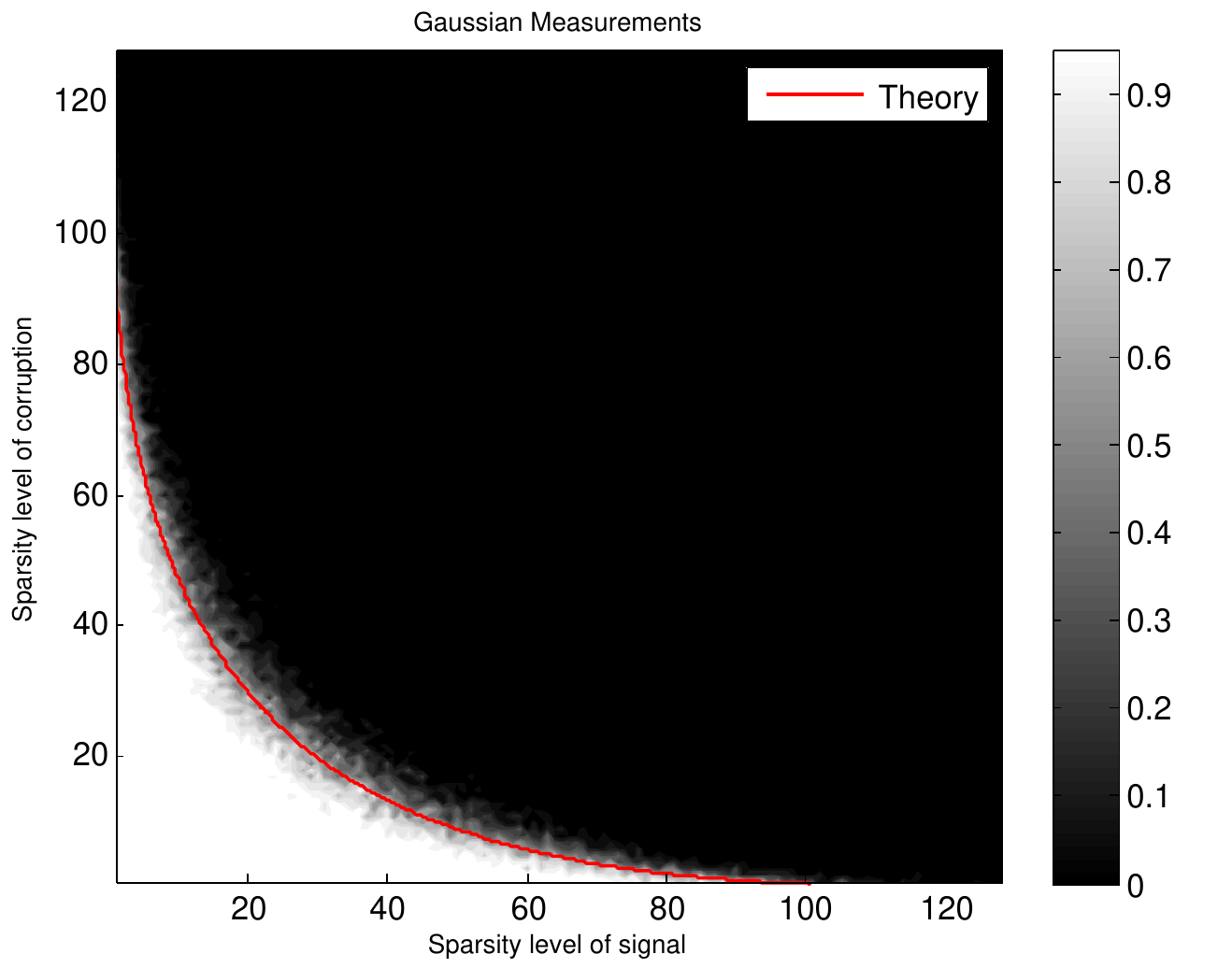}
	}
	\subfigure{
  		\includegraphics[width= .48\textwidth]{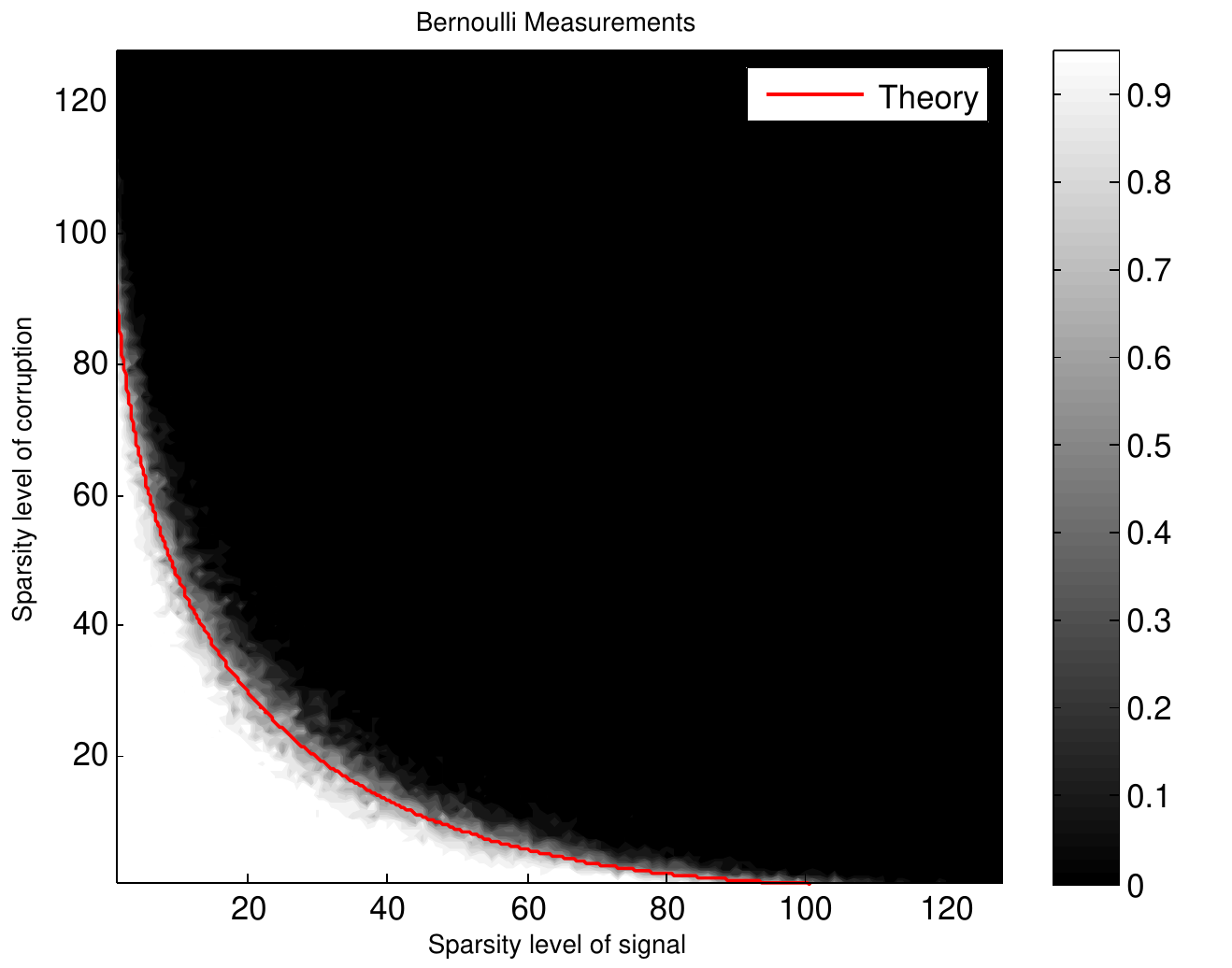}
	}
	\caption{Phase transition of the constrained recovery procedure \eqref{L1constrained} under Gaussian and Bernoulli measurements. The red curve plots the phase transition threshold predicted by \eqref{threshold}.}
	\label{fig:Constrained_Procedure}
\end{figure*}

\subsection{Phase Transition of Partially Penalized Recovery Procedure}
We next consider the partially penalized recovery procedure in the noiseless setting where neither $f(\vx)=\|\vx^{\star}\|_1$ nor $g(\vv)=\|\vv^{\star}\|_1$ is known beforehand. The experiment settings are almost the same as in the constrained case except that we recover signals and corruptions via the following procedure in step (4):
\begin{equation}\label{L1PartiallyPenalized}
\min_{\vx,\vv}  \|\vx\|_1 + \lambda \|\vv\|_1, \text{ s.t. }\vy = \mPhi \vx +\vv.
\end{equation}
We test two kinds of the regularization parameter: 1) $\lambda = \lambda^{\star}$, which is chosen according to \eqref{Chooselambda} and depends on the signal and corruption sparsity levels $s_{\sig}$ and $s_{\cor}$, and 2) $\lambda = 1$. As discussed in Section \ref{Relationships}, when we choose $\lambda = \lambda^{\star}$, the necessary number of measurements to guarantee success by the partially penalized procedure is nearly the same as that of the constrained one. Thus, we also overlay the curve of \eqref{threshold} to compare with the empirical results.

Fig. \ref{fig:Partially_Penalized_Procedure} shows the empirical probability of success as the signal and corruption sparsity levels $s_{\sig}$ and $s_{\cor}$ vary. We can find that the recovery performance in the case of $\lambda = \lambda^{\star}$ is better than that in the case of $\lambda = 1$. This demonstrates that the choice of the regularization parameter suggested by \eqref{Chooselambda} is effective. Moreover, signal recovery with $\lambda = \lambda^{\star}$ is nearly as good as in the constrained case, but without any prior knowledge of $\|\vx^{\star}\|_1$.

\begin{figure*}
  	\centering
	\subfigure{
  		\includegraphics[width= .48\textwidth]{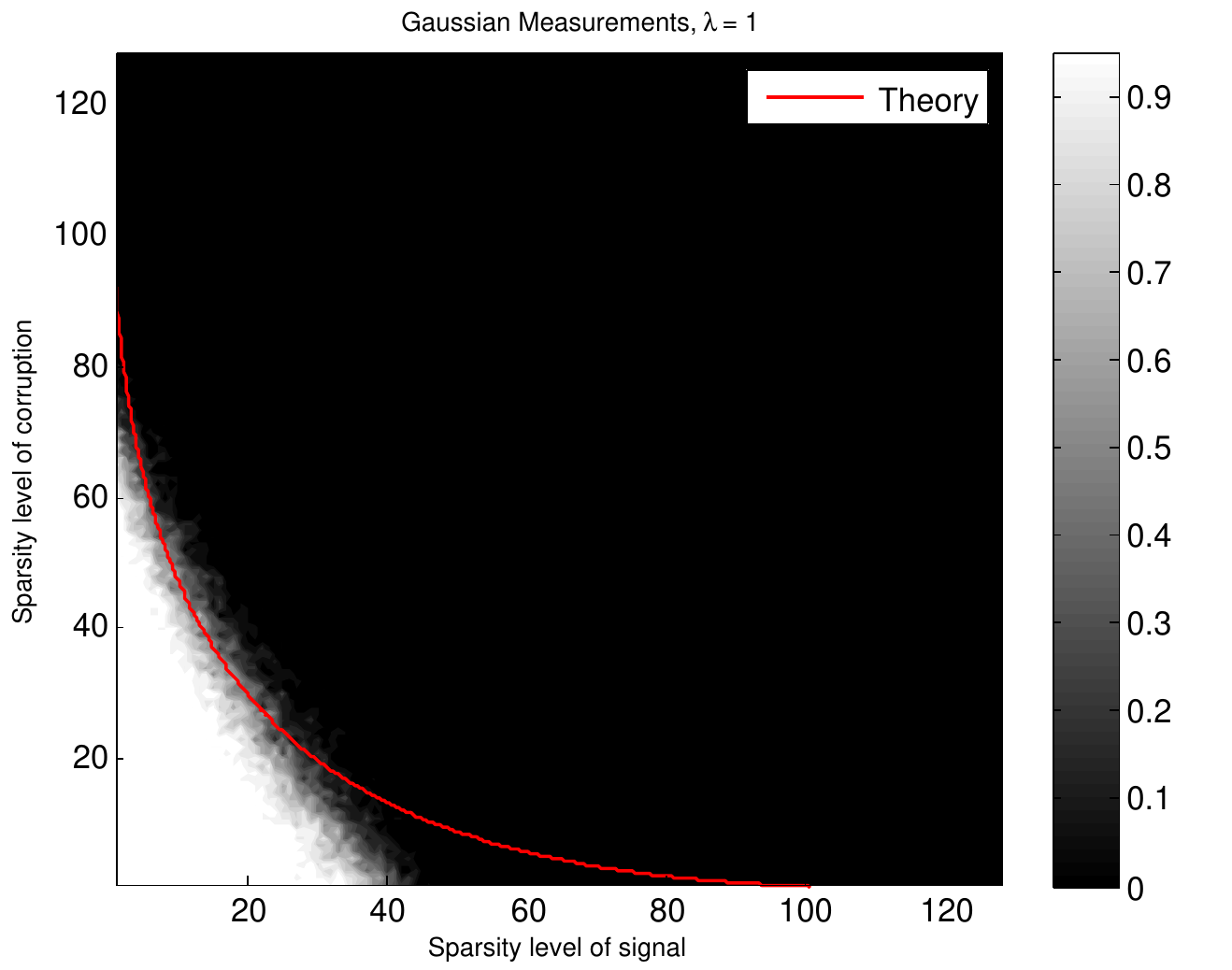}
        \includegraphics[width= .48\textwidth]{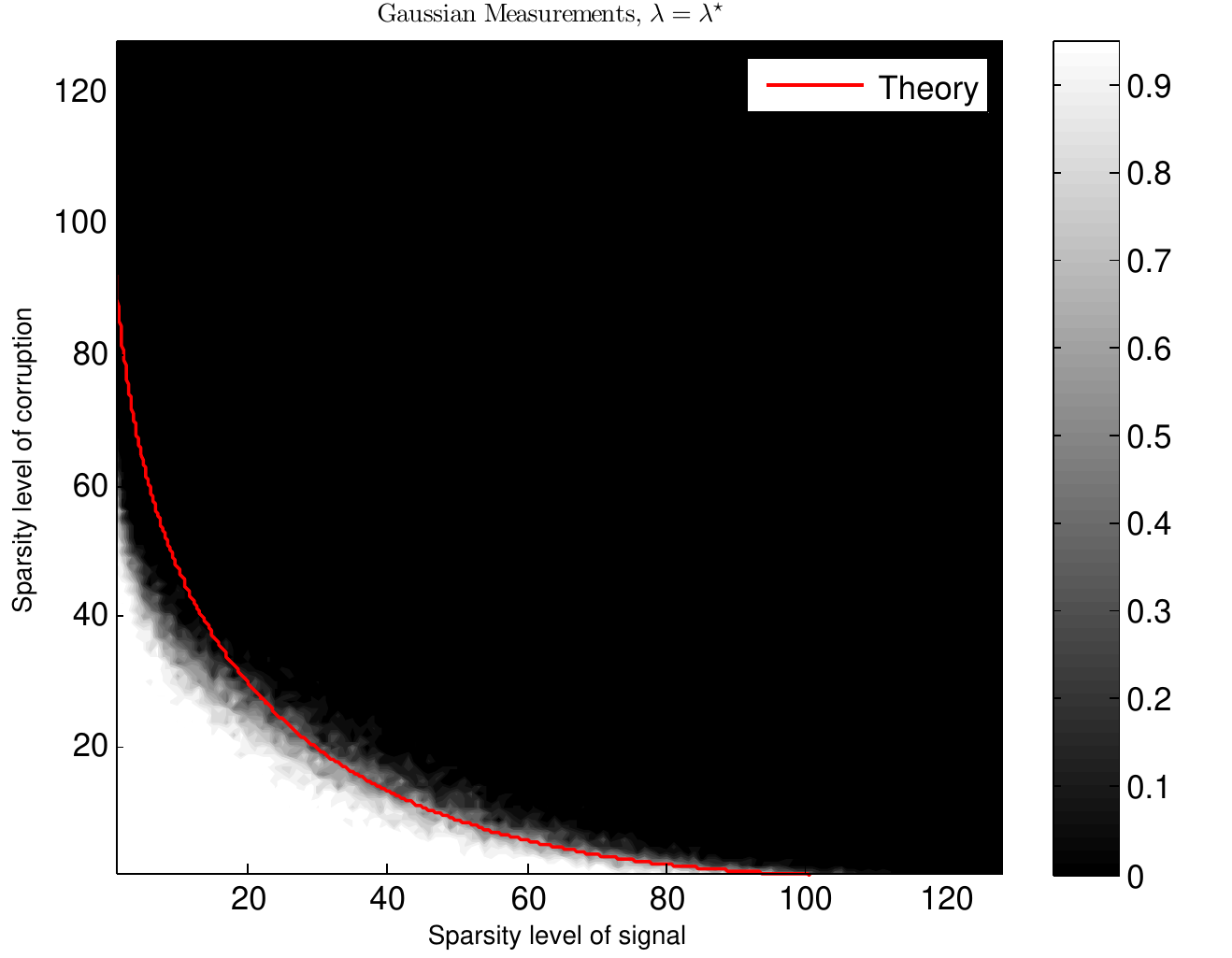}
	}
	\subfigure{
  		\includegraphics[width= .48\textwidth]{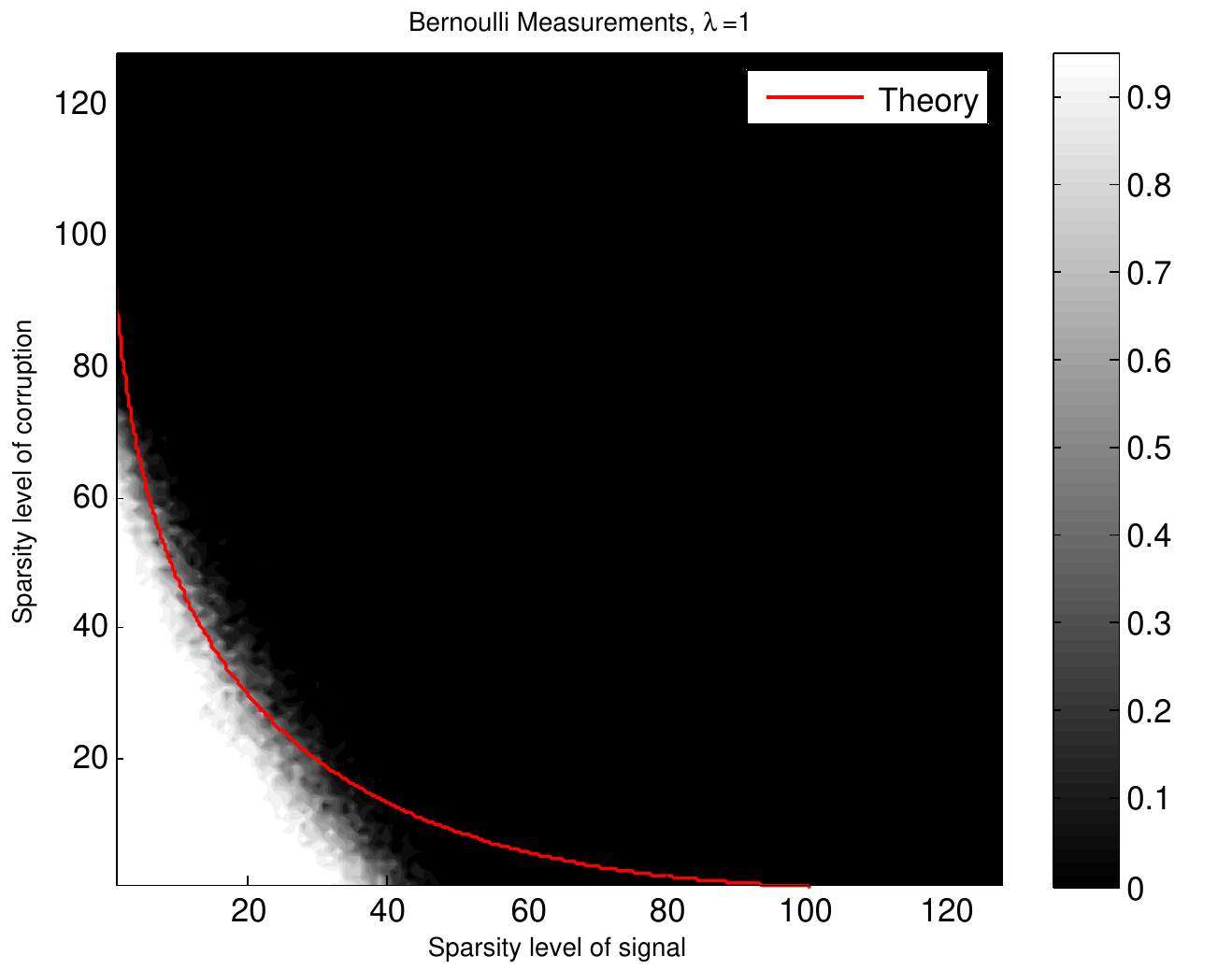}
        \includegraphics[width= .48\textwidth]{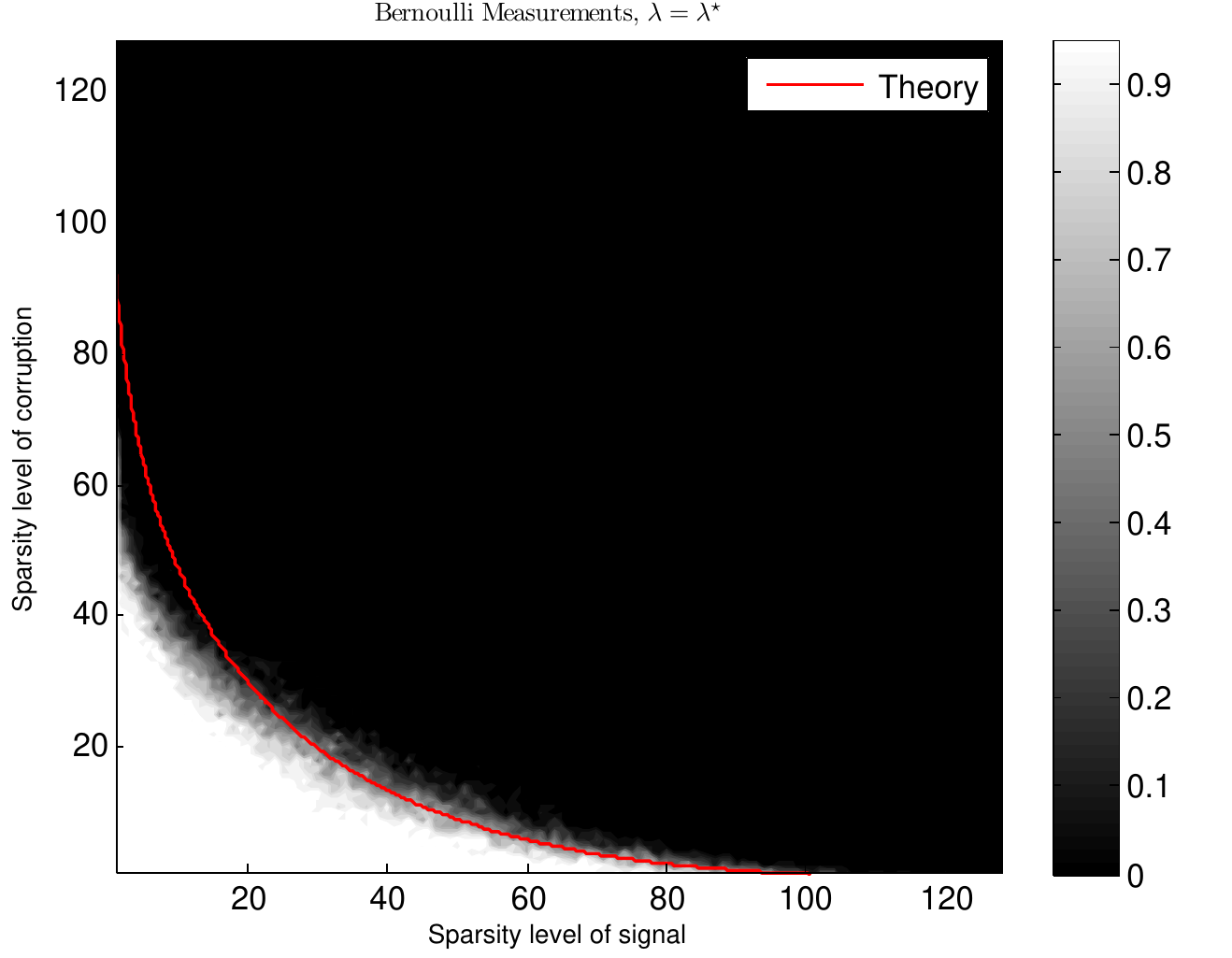}
	}
	\caption{Phase transition of the partially penalized recovery procedure \eqref{L1PartiallyPenalized} with different kinds of the regularization parameter under Gaussian and Bernoulli measurements. The red curve plots the phase transition threshold predicted by \eqref{threshold}.}
	\label{fig:Partially_Penalized_Procedure}
\end{figure*}

\subsection{Phase Transition of Fully Penalized Recovery Procedure}
Thirdly, we study the empirical behavior of the fully penalized recovery procedure in the noiseless setting where neither $f(\vx)=\|\vx^{\star}\|_1$ nor $g(\vv)=\|\vv^{\star}\|_1$ is known a priori. Similarly, the experiment settings are nearly the same as in the constrained case except that we reconstruct signals and corruptions via the following procedure in step (4):
\begin{equation}\label{L1FullyPenalized}
\min_{\vx,\vv}   \frac{1}{2} \|\vy - \mPhi \vx -\vv\|_2^2 + \tau_1 \|\vx\|_1 + \tau_2 \|\vv\|_1.
\end{equation}
As suggested in Section \ref{Howtochoosetau}, in order to recover signals and corruptions faithfully by this procedure, we require that the regularization parameters $\tau_1$ and $\tau_2$ tend to zero. Thus we set $\tau_1 = \tau_2 = 10^{-5}$. For reference, we also overlay the curve of \eqref{threshold} to compare with the empirical results.

Fig. \ref{fig:Fully_Penalized_Procedure} illustrates the average empirical probability of success for each setting $(s_{\sig}, s_{\cor})$. We can see that the theoretical recovery threshold \eqref{threshold} roughly predict observed phase transition with this choice of the regularization parameters.

\begin{figure*}
  	\centering
	\subfigure{
  		\includegraphics[width= .48\textwidth]{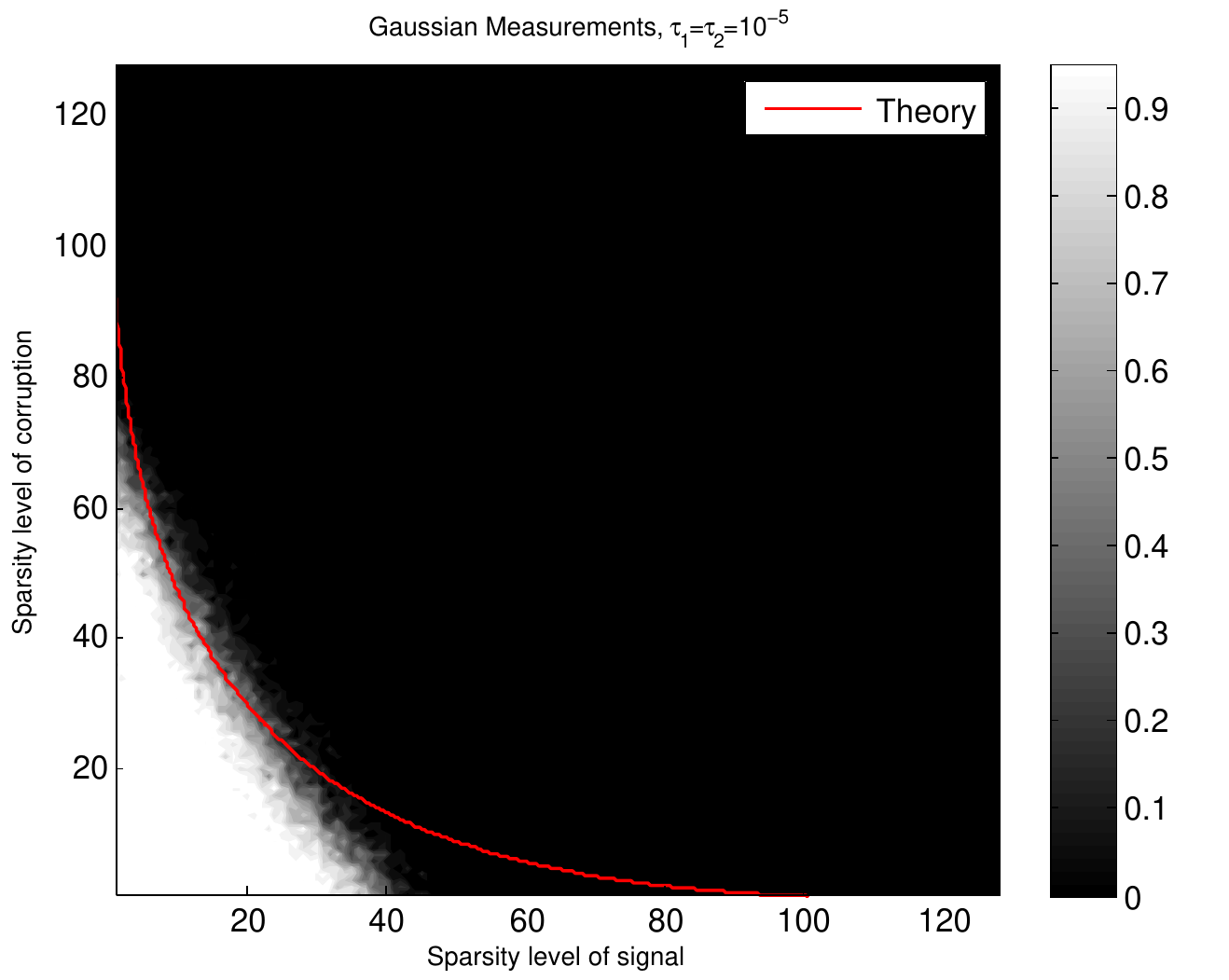}
	}
	\subfigure{
  		\includegraphics[width= .48\textwidth]{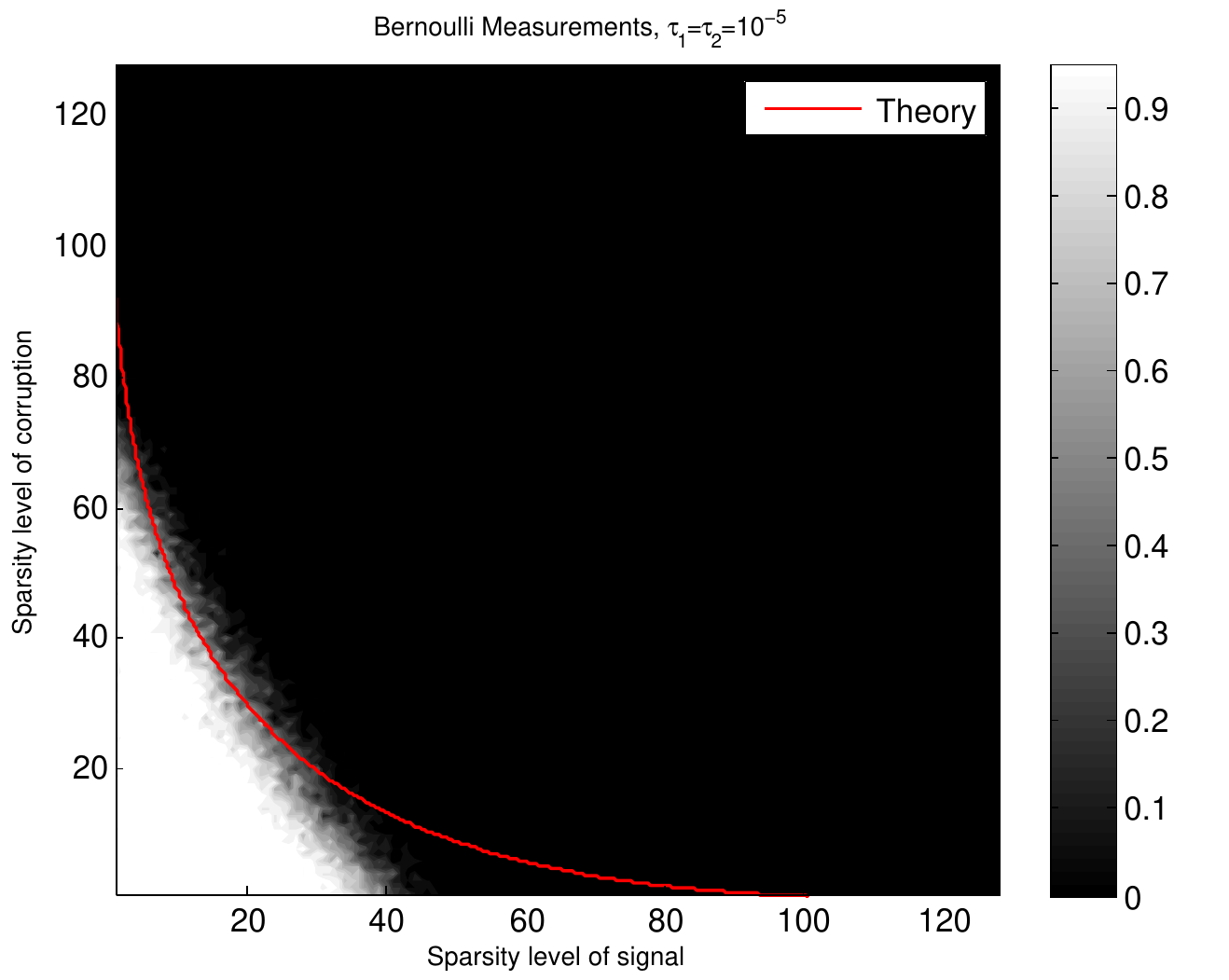}
	}
	\caption{Phase transition of the fully penalized recovery procedure \eqref{L1FullyPenalized} with $\tau_1 = \tau_2 = 10^{-5}$ under Gaussian and Bernoulli measurements. The red curve plots the phase transition threshold predicted by \eqref{threshold}.}
	\label{fig:Fully_Penalized_Procedure}
\end{figure*}

\subsection{Stable Recovery}
Finally, we investigate the empirical behavior of the three procedures under noisy measurements. We fix $m = n = 128$  and  $s_{\sig} = s_{\cor} = 20$. The first three steps of this experiment are the same as that in the constrained case and the other steps are as follows:
\begin{itemize}
	\item [(4)] Generate a Gaussian noise vector $\vz$ and scale $\vz$ such that $\|\vz\|_2 = \delta$.

    \item [(5)] Solve constrained, partially penalized, and fully penalized procedures under different noise levels with different regularization parameters. In the partially penalized case, we choose $\lambda = \lambda^{\star}$ according to \eqref{Chooselambda} or $\lambda = 1$.
        In the fully penalized case, we choose $\tau_1$ and $\tau_2$ according \eqref{ChoosetauCase2} or $\tau_1 = \tau_2 = 1$. Here, for sparse signal and sparse corruption, \eqref{ChoosetauCase2} implies $\tau_1 = \tau_{1B} \approx CK\beta\delta(\sqrt{n}+\sqrt{m})/\sqrt{m} =  2CK\beta\delta = C'K\beta\delta$ and $\tau_2 = \tau_{2B} = \beta \delta$. In our simulations, we choose $C'K=1$ and $\beta = 2$.

    \item [(6)] Record the $\ell_2$ recovery error $\sqrt{\|\hat{\vx}-\vx^{\star}\|_2^2 + \|\hat{\vv}-\vv^{\star}\|_2^2  }$ in each case.
\end{itemize}

Fig. \ref{fig:Stable_recovery} shows the average recovery error across $20$ runs. We can see that the recovery error increases linearly as the noise level increases in both constrained and partially penalized cases, which is consistent with our theoretical results (Theorems \ref{them: Constrained Recovery} and \ref{them: Partially_Penalized_Recovery}). Moreover, the partially penalized procedure with different regularization parameters shows almost the same performance as the constrained one when the sparsity levels of signal and corruption are not too high ($s_{\sig} = s_{\cor} = 20$).
This implies that when the estimate of the sparsity levels $s_{\sig}$ and $s_{\cor}$ is unavailable, the setting $\lambda = 1$ yields high probability recovery, provided that the sparsity levels of signal and corruption are relatively low. In the fully penalized case, the recovery performance with $\tau_1$ and $\tau_2$ chosen according to \eqref{ChoosetauCase2} is much better than that with $\tau_1 = \tau_2 =1$. This demonstrates the effectiveness of the parameter selection strategy \eqref{ChoosetauCase2}.

\begin{figure*}
  	\centering
	\subfigure{
  		\includegraphics[width= .48\textwidth]{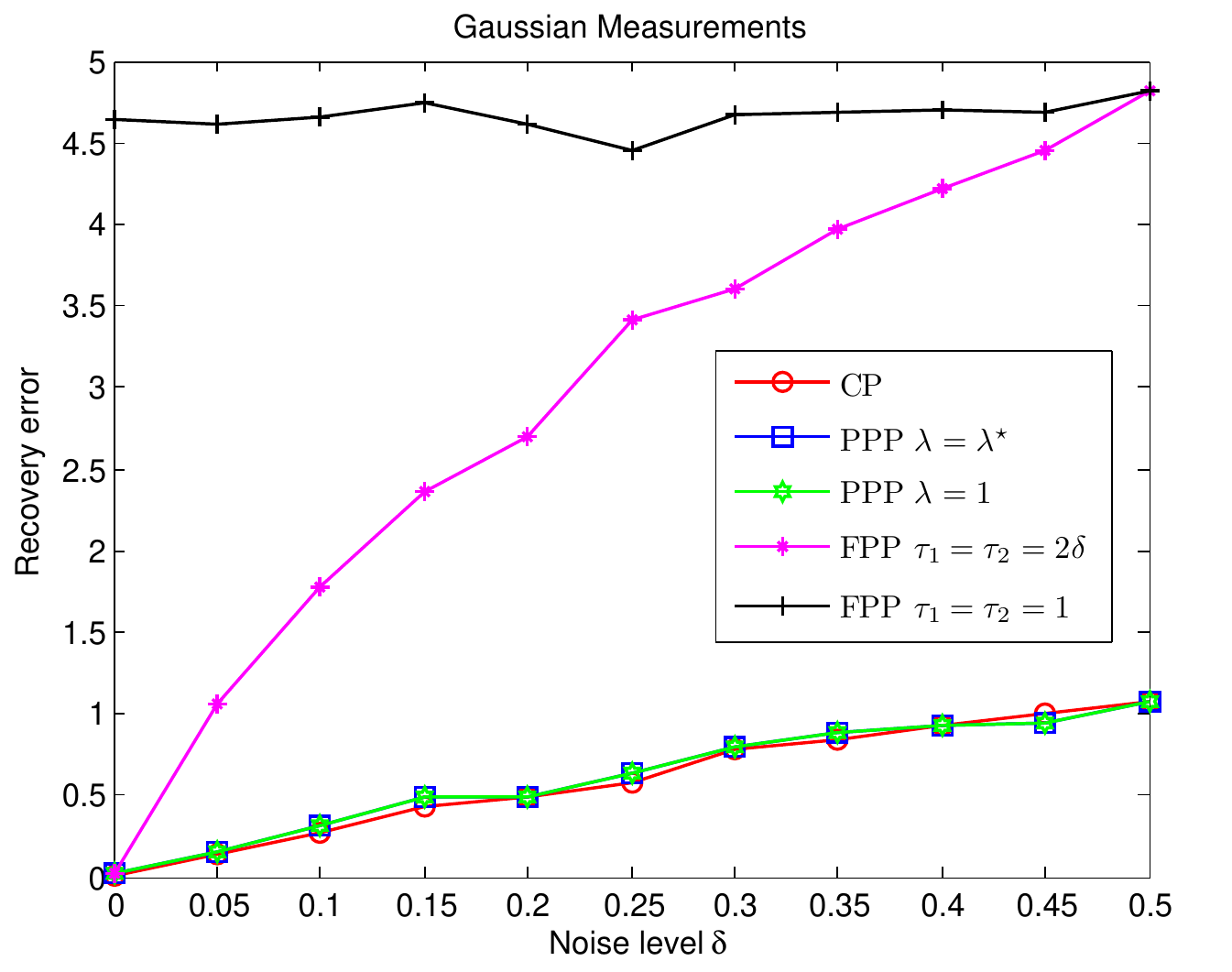}
	}
	\subfigure{
  		\includegraphics[width= .48\textwidth]{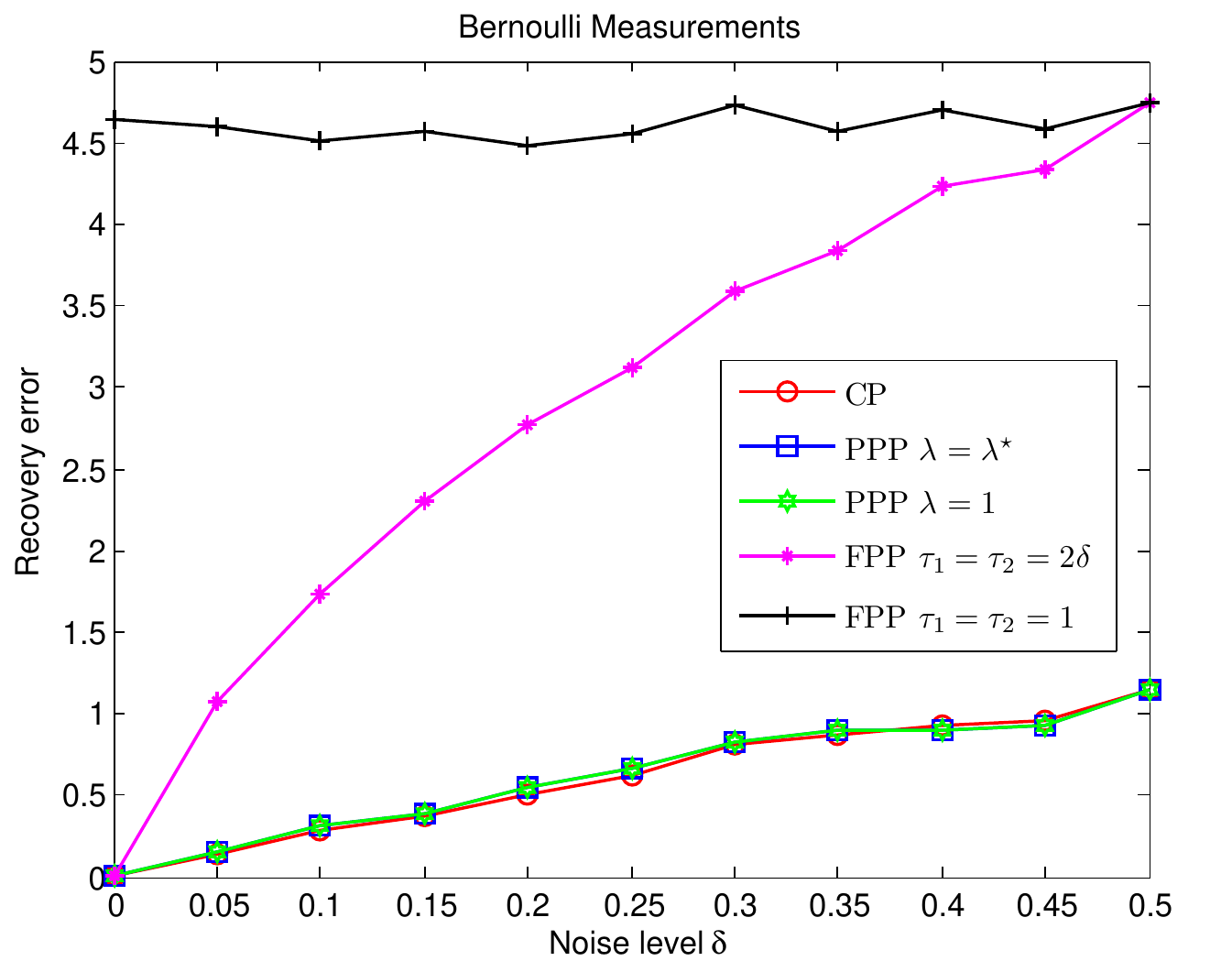}
	}
	\caption{Recovery error $\sqrt{\|\hat{\vx}-\vx^{\star}\|_2^2 + \|\hat{\vv}-\vv^{\star}\|_2^2}$ by constrained, partially penalized, and fully penalized procedures (abbreviated as CP, PPP, and FPP, respectively) under different noise levels with different regularization parameters. Left: Gaussian measurements, right: Bernoulli measurements.}
	\label{fig:Stable_recovery}
\end{figure*}

\section{Conclusion}\label{Conclusion}
In this paper, we have established an extended matrix deviation inequality for sub-Gaussian matrices, which provides a powerful tool to analyze the corrupted sensing problem. We then presented performance analysis for three convex recovery procedures which are used to recover structured signals from corrupted sub-Gaussian measurements when different kinds of prior information are available. We considered both bounded and stochastic noise. Our results have shown that, under proper conditions, these approaches reconstruct both signal and corruption exactly in the absence of noise and stably in the presence of noise. Moreover, our results also indicate how to pick the regularization parameters in both partially and fully penalized recovery procedures and reveal the relationships among these methods. For future work, it would be of great interest to establish the phase transition theory for both constrained and partially penalized procedures under sub-Gaussian measurements.

\appendices

\section{Some Properties of Sub-Gaussian and Sub-exponential random variables and vectors}
%

\begin{fact}[Product of sub-Gaussians is sub-exponential]\cite[Lemma 2.7.7]{vershynin2016book} \label{Product of subgaus is subexp}
	Let $X$ and $Y$ be sub-Gaussian random variables (not necessarily independent). Then $XY$ is sub-exponential. Moreover,
	\begin{equation}
	\|XY\|_{\psi_1} \leq \|X\|_{\psi_2}\|Y\|_{\psi_2}.
	\end{equation}
\end{fact}

\begin{fact}[Centering]\cite[Lemma 2.6.8 and Exercise 2.7.10]{vershynin2016book}
	\label{pro: center}
	If $X$ is sub-Gaussian (or sub-exponential), then so is $X-\E X$. Moreover,
	\begin{equation*}
	\|X - \E X\|_{\psi_2} \leq C \|X\|_{\psi_2} ~~\textrm{and}~~ \quad \|X-\E X\|_{\psi_1} \leq C \|X\|_{\psi_1}.
	\end{equation*}
\end{fact}

\begin{fact}[Hoeffding-type Inquality] \cite[Theorem 2.6.3]{vershynin2016book}
	\label{Hoeffding's ineq}
	Let $X_1, \ldots, X_m$ be independent, mean-zero, sub-Gaussian random variables, and $\va = (a_1, a_2, \ldots, a_m)^{T}\in\R^{m}$. Then, for any $t\geq 0$, we have
	\begin{equation}
	\Pr{\left| \sum_{i=1}^{m}a_i X_i \right| \geq t } \leq 2 \exp\left\{-\frac{ct^2}{K^2 \|\va\|_2^2} \right\},
	\end{equation}
	where $K = \max_{i} \|X_i\|_{\psi_2}$.
\end{fact}

\begin{fact}[Bernstein-type Inequality] \cite[Theorem 2.8.2]{vershynin2016book}
    \label{Bernstein ineq}
	Let $X_1, X_2, \ldots, X_m $ be independent, mean-zero, sub-exponential random variables, and $\va = (a_1, a_2, \ldots, a_m)^{T}\in\R^{m}$. Then, for any $t \geq 0$, we have
	\begin{align}
	\Pr{ \left| \sum_{i=1}^{m} a_{i} X_{i} \right| \geq t }\leq 2 \exp \left\{ -c \min  \left( \frac{t^2}{K^2 \|\va\|_2^2}, \frac{t}{K\|\va\|_{\infty} } \right)\right\},
	\end{align}
	where $K = \max_{i} \|X_{i}\|_{\psi_1}$.
\end{fact}

\section{Proof of Theorem \ref{them: mat dev ineq}}
\label{ProofofEMDI}
In order to prove Theorem \ref{them: mat dev ineq}, we need the following key lemma which states that the random process $X_{\va,\vb}:= \|\mA\va + \sqrt{m}\vb\|_2 - \sqrt{m}\cdot\sqrt{\|\va\|^2_2 + \|\vb\|_2^2}$ has sub-Gaussian increments.
\begin{lemma}[Sub-Gaussian Process]
	\label{them: sub-Gaussian process}
	Let $\mA$ be an $m \times n$ matrix whose rows $\mA_i$ are independent centered isotropic sub-Gaussian vectors. Then the random process
	\begin{align*}
	X_{\va,\vb} &=\|\mA\va+\sqrt{m}\vb\|_2 - (\E\|\mA\va + \sqrt{m}\vb\|_2^2)^{1/2}\\
	&=\|\mA\va + \sqrt{m}\vb\|_2 - \sqrt{m}\cdot\sqrt{\|\va\|^2_2 + \|\vb\|_2^2}
	\end{align*}
    has sub-Gaussian increments:
    \begin{equation}\label{subGaussian_increments}
	\left\| X_{\va,\vb} - X_{\va',\vb'} \right\|_{\psi_2} \leq CK^2\cdot\sqrt{\|\va-\va'\|_2^2 + \|\vb-\vb'\|_2^2} ~~~~\text{for every}~~(\va,\vb),(\va',\vb')\in\R^n\times\R^m,
	\end{equation}
    where $K = \max_{i}\|\mA_i\|_{\psi_2}$.
\end{lemma}

\begin{proof}
  See Appendix \ref{ProofofSGP}.
\end{proof}

Combing Lemma \ref{them: sub-Gaussian process} and Talagrand's Majorizing Measure Theorem yields the proof of Theorem \ref{them: mat dev ineq}.

\begin{proof}
According to Lemma \ref{them: sub-Gaussian process}, the random process
\begin{align*}
X_{\va,\vb}:=\|\mA\va + \sqrt{m}\vb\|_2 - \sqrt{m}\cdot\sqrt{\|\va\|_2^2 + \|\vb\|_2^2}
\end{align*}
has sub-Gaussian increments, that is
\begin{align*}
\|X_{\va,\vb}-X_{\va',\vb'}\|_{\psi_2}&\leq CK^2\cdot\sqrt{\|\va-\va'\|_2^2 + \|\vb-\vb'\|_2^2}\\
&=CK^2\left\|
	\begin{bmatrix}
		\va\\
		\vb\\
	\end{bmatrix}
   -\begin{bmatrix}
		\va'\\
		\vb'\\
    \end{bmatrix}
\right\|_2
\end{align*}
for every $(\va,\vb), (\va',\vb') \in \TT$. It follows from Talagrand's Majorizing Measure Theorem that
\begin{align*}
\E\sup_{(\va,\vb),(\va',\vb')\in \TT}\big|X_{\va,\vb}-X_{\va',\vb'}\big|&\leq CK^2\E\sup_{(\va,\vb)\in \TT}\ip{\begin{bmatrix}
	\vg\\
	\vh\\
	\end{bmatrix}}{\begin{bmatrix}
	\va\\
	\vb\\
	\end{bmatrix}}\\
&=CK^2\omega(\TT),
\end{align*}
where $\vg\sim\NN(\vzero, \mI_n)$ and $\vh\sim\NN(\vzero, \mI_m)$. Thus, fix an arbitrary point $(\va_0,\vb_0)\in \TT$ and use the triangle inequality to obtain
\begin{align*}
\E\sup_{(\va,\vb)\in \TT}\big|X_{\va,\vb} \big| & \leq \E\sup_{(\va,\vb)\in \TT}\big|X_{\va,\vb} - X_{\va_0,\vb_0} \big| + \E\big| X_{\va_0,\vb_0} \big|\\
                                              & \leq CK^2 \omega(\TT) +  \E\big| X_{\va_0,\vb_0} \big|.
\end{align*}
The second term can be bounded as follows:
\begin{align*}
\E\big| X_{\va_0,\vb_0} \big| \leq C'\| X_{\va_0,\vb_0} \|_{\psi_2}  \leq C''K^2\cdot\left\|\begin{bmatrix}
\va_0\\
\vb_0\\
\end{bmatrix}\right\|_2.
\end{align*}
The first inequality follows from the equivalent definition \eqref{Sub-Gaussian_Definition1} and the second inequality holds by using Lemma \ref{them: sub-Gaussian process} with $[\va'^T,\vb'^T]^T = \vzero$.
Therefore, by \eqref{Relation}, we have
\begin{align*}
\E\sup_{(\va,\vb)\in \TT}\big|X_{\va,\vb} \big|&\leq CK^2 \omega(\TT) + C''K^2\cdot\left\|\begin{bmatrix}
\va_0\\
\vb_0\\
\end{bmatrix}\right\|_2\\
&\leq C_0K^2\gamma(\TT).
\end{align*}
This establishes the expectation bound.

For the high probability bound, define $\bar{\TT} = \TT \cup \{\vzero\}$. It also follows from Talagrand's Majorizing Measure Theorem that with probability at least $1-\exp(-t^2)$,
\begin{align*}
\sup_{(\va,\vb) \in \TT}\Big| X_{\va,\vb} \Big| = \sup_{(\va,\vb) \in \bar{\TT}}\Big| X_{\va,\vb} \Big| & = \sup_{(\va,\vb) \in \bar{\TT}}\Big| X_{\va,\vb} - X_{\vzero,\vzero} \Big| \\
                                                    & \leq \sup_{(\va,\vb),(\va',\vb') \in \bar{\TT}}\Big| X_{\va,\vb} - X_{\va',\vb'} \Big| \\
                                                    &\leq C'K^2 \big[\omega(\bar{\TT})+ t\cdot\diam(\bar{\TT})\big]\\
                                                    &\leq CK^2 \big[\gamma(\TT)+ t\cdot\rad(\TT)\big].
\end{align*}
In the last inequality, we have used the facts that $\omega(\bar{\TT}) \leq \gamma(\bar{\TT}) = \gamma(\TT)$ and $\diam(\bar{\TT})\leq 2 \rad(\TT)$.
This completes the proof.
\end{proof}

\section{Proof of Lemma \ref{them: sub-Gaussian process}}
\label{ProofofSGP}
Our proof of Lemma \ref{them: sub-Gaussian process} is inspired by \cite{Schechtman2006} and \cite{liaw2016simple}. For clarity, the proof is divided into three steps. First, we show a partial case of Lemma \ref{them: sub-Gaussian process} in which $(\va^T,\vb^T)^T\in\S^{n+m-1}$ and $(\va'^T,\vb'^T)^T= \vzero$. We then extend this by allowing $(\va'^T,\vb'^T)^T$ to be an arbitrary unit vector. Finally, we prove the increment inequality \eqref{subGaussian_increments} for any $(\va,\vb), (\va',\vb')\in\R^n\times\R^m$.

Since $\mA_i$ are isotropic sub-Gaussian random vectors, it follows from \eqref{Sub-Gaussian_Definition1} that $K$ is bounded below by an absolute constant. For simplicity, we will assume that $K \geq 1$.

\noindent\textbf{Step 1:} $(\va^T,\vb^T)^T\in\S^{n+m-1}$ and $(\va'^T,\vb'^T)^T= \vzero$.
In this case, we have the following result.
\begin{lemma}\label{lemma: concen of cor random mat}
	Let $\mA$ be a sub-Gaussian random matrix as in Lemma \ref{them: sub-Gaussian process}. Then
	\begin{equation}
	\Big\|\|\mA\va + \sqrt{m}\vb\|_2 - \sqrt{m}\Big\|_{\psi_2} \leq CK^2 \quad\text{for every}\quad[\va^T,\vb^T]^T\in\S^{n+m-1}.
	\end{equation}
\end{lemma}

\begin{proof}
We begin with establishing a concentration inequality for $\frac{1}{m}\|\mA\va +\sqrt{m} \vb\|_2^2 - 1$. For any $t\geq 0$, we have
\begin{align}\label{Bound_p}
 p &:=\Pr{ \left| \frac{1}{m} \|\mA\va + \sqrt{m}\vb\|_2^2 - 1 \right| \geq t }\\ \notag
   &=\Pr{\left| \frac{1}{m} \|\mA\va \|_2^2 + \frac{2}{\sqrt{m}}\ip{\mA\va}{\vb} + \|\vb\|_2^2 - 1 \right| \geq t}\\ \notag
   &=\Pr{ \left| \frac{1}{m}\|\mA\va\|_2^2 -\|\va\|_2^2 + \frac{2}{\sqrt{m}}\ip{\mA\va}{\vb} \right| \geq t }\\ \notag
   &\leq\Pr{ \left| \frac{1}{m}\|\mA\va\|_2^2 -\|\va\|_2^2\right| + \left|\frac{2}{\sqrt{m}}\ip{\mA\va}{\vb} \right| \geq t }\\ \notag
   &\leq\Pr{ \left| \frac{1}{m} \|\mA\va\|_2^2 -\|\va\|_2^2 \right| \geq \frac{t}{2} } +\Pr{ \left| \frac{2}{\sqrt{m}}\ip{\mA\va}{\vb} \right| \geq \frac{t}{2} }\\ \notag
   &=:p_1+p_2.
\end{align}

To bound $p_1$, it will be useful to express $\frac{1}{m} \|\mA\va\|_2^2 -\|\va\|_2^2$ as a sum of independent random variables $\frac{1}{m}\sum_{i=1}^{m}[\ip{\mA_i^T}{\va}^2 - \|\va\|_2^2]:=\frac{1}{m}\sum_{i=1}^{m}Z_i$. By assumption, $\{\ip{\mA_i^T}{\va}\}_{i=1}^m$ are independent centered sub-Gaussian random variables with $\E\ip{\mA_i^T}{\va}^2 = \|\va\|_2^2$ and $\|\ip{\mA_i^T}{\va}\|_{\psi_2}\leq K\|\va\|_2\leq K$. Therefore, by facts \ref{pro: center} and \ref{Product of subgaus is subexp} , $\{Z_i\}_{i=1}^m$ are independent centered sub-exponential random variables with
\begin{align*}
 \|Z_i\|_{\psi_1}&=\|\ip{\mA_i^T}{\va}^2 - \E \ip{\mA_i^T}{\va}^2 \|_{\psi_1}\\
 &\leq C\| \ip{\mA_i^T}{\va}^2 \|_{\psi_1}~~ \text{by Fact }\ref{pro: center}\\
 &\leq C\|\ip{\mA_i^T}{\va}\|_{\psi_2}^2~~~\text{by Fact }\ref{Product of subgaus is subexp}\\
 &\leq CK^2.
\end{align*}
It follows from Bernstein's inequality (Fact \ref{Bernstein ineq}) that
\begin{align}\label{Bound_p1}
p_1=&\Pr{ \left| \frac{1}{m}\sum_{i=1}^{m} Z_i \right| \geq \frac{t}{2} } \\ \notag
\leq&2\exp\left\{ -c\min ( \frac{t^2}{4C^2K^4/m}, \frac{t}{2CK^2/m} ) \right\}\\ \notag
  = &2\exp\left\{ -\frac{cm}{4C^2K^4}\min(t^2, 2CK^2t) \right\}\\ \notag
\leq&2\exp\left\{ -\frac{c_1m}{K^4}\min(t^2, t) \right\}.
\end{align}
The last inequality holds because we can easily choose $C$ such that $2CK^2 \geq 1$.

To bound $p_2$, it will be helpful to write $\ip{\mA\va}{\vb}=\sum_{i=1}^{m}\vb_i\ip{\mA_i^T}{\va}$, where $\{\ip{\mA_i^T}{\va}\}_{i=1}^m$ are independent centered sub-Gaussian random variable with  $\|\ip{\mA_i^T}{\va}\|_{\psi_2}\leq K$. Applying Hoeffding's inequality (Fact \ref{Hoeffding's ineq}) yields
\begin{align}\label{Bound_p2}
p_2 &= \Pr{ \left| \frac{1}{\sqrt{m}}\sum_{i=1}^{m}\vb_i\ip{\mA_i^T}{\va} \right| \geq \frac{t}{4} }\\ \notag
&\leq 2\exp\left\{ -\frac{cmt^2}{16K^2\|\vb\|_2^2} \right\}\\ \notag
&\leq 2\exp\left\{ -\frac{c_2mt^2}{K^4}\right\}.
\end{align}
The last inequality holds because $\|\vb\|_2 \leq 1$ and $K \geq 1$.

Combining the upper bounds of $p_1$ and $p_2$ yields
\begin{align} \label{concentrationinequalityLemma}
p&\leq 2\exp\left\{ -\frac{c_1m}{K^4}\min(t^2, t) \right\} + 2\exp\left\{ -\frac{c_2mt^2}{K^4}\right\}\\ \notag
 &\leq 4\exp\left\{ -\frac{c_0m}{K^4}\min(t^2, t) \right\}.
\end{align}

We then establish a concentration inequality for $\frac{1}{\sqrt{m}}\|\mA\va + \sqrt{m}\vb\| - 1$. To this end, we will use the following fact
\begin{equation*}
|z-1|\geq \delta \quad\text{implies}\quad |z^2-1|\geq \max(\delta^2, \delta), ~~\textrm{for any} ~~z\geq 0~~\textrm{ and}~~ \delta\geq 0.
\end{equation*}
Using this for $z = \frac{1}{\sqrt{m}}\|\mA\va + \sqrt{m}\vb\|$ together with \eqref{concentrationinequalityLemma} where $t = \max\{\delta^2, \delta\}$, we obtain for any $\delta \geq 0$ that
\begin{align*}
	\Pr{ \left| \frac{1}{\sqrt{m}} \|\mA\va + \sqrt{m}\vb\|_2 - 1 \right| \geq \delta } \leq & \Pr{ \left| \frac{1}{m} \|\mA\va + \sqrt{m}\vb\|_2^2 - 1 \right|  \geq \max(\delta^2, \delta) }\\
    \leq& 4\exp\left\{ -\frac{c_0m\delta^2}{K^4}\right\}.
\end{align*}
This implies $\big\| \| \mA\va + \sqrt{m}\vb \|_2 - \sqrt{m}\big\|_{\psi_2} \leq CK^2$ and completes the proof.
\end{proof}

\noindent\textbf{Step 2:} $(\va^T,\vb^T)^T\in\S^{n+m-1}$ and $(\va'^T,\vb'^T)^T\in\S^{n+m-1}$. In this case, we have the following result.
\begin{lemma}
	\label{lemma: double units}
	Let $A$ be a sub-Gaussian random matrix as in Lemma \ref{them: sub-Gaussian process}. Then
	\begin{align*}
	\Big\| \|\mA\va + \sqrt{m}\vb\|_2 - \|\mA\va' + \sqrt{m}\vb'\|_2 \Big\|_{\psi_2}\leq CK^2\cdot\sqrt{\|\va-\va'\|_2^2 + \|\vb-\vb'\|_2^2},
	\end{align*}
	for every $(\va^T,\vb^T)^T, (\va'^T,\vb'^T)^T\in\S^{n+m-1}$.
\end{lemma}
\begin{proof}
Given $t\geq 0$, it suffices to show
\begin{align}
\label{key ineq unit vectors}
p:=\Pr{\frac{\big| \|\mA\va + \sqrt{m}\vb\|_2 - \|\mA\va' + \sqrt{m}\vb'\|_2 \big|}{\sqrt{\|\va-\va'\|_2^2 + \|\vb-\vb'\|_2^2}} \geq t } \leq C \exp\left\{\frac{-ct^2}{K^4}\right\}.
\end{align}
We will proceed differently for small and large $t$.

\textbf{Case 1:} $t\geq 2\sqrt{m}$. Denote
	\begin{equation*}
	\vu:=\frac{ \va-\va' }{\sqrt{\|\va-\va'\|_2^2 + \|\vb-\vb'\|_2^2}}~~\textrm{and}~~\quad\vv:=\frac{ \vb-\vb' }{\sqrt{\|\va-\va'\|_2^2 + \|\vb-\vb'\|_2^2}}.
	\end{equation*}
It follows from the triangle inequality that
	\begin{align*}
	p&\leq \Pr{\frac{ \|\mA(\va-\va') + \sqrt{m}(\vb-\vb')\|_2 }{ \sqrt{\|\va-\va'\|_2^2 + \|\vb-\vb'\|_2^2} } \geq t }\\
	 &=\Pr{\|\mA\vu + \sqrt{m}\vv\|_2\geq t}\\
	 &=\Pr{\|\mA\vu + \sqrt{m}\vv\|_2 - \sqrt{m}\geq t-\sqrt{m}}\\
	 &\leq\Pr{\|\mA\vu + \sqrt{m}\vv\|_2 - \sqrt{m}\geq \frac{t}{2}}\\
	 &\leq 2\exp\left\{-\frac{ct^2}{K^4}\right\}.
	\end{align*}
The second inequality holds because $t\geq 2\sqrt{m}$  and the last inequality follows from Lemma \ref{lemma: concen of cor random mat}.


\textbf{Case 2:} $t\leq 2\sqrt{m}$. Denote
\begin{equation*}
\vu' = \va+\va' \text{ and } \vv' = \vb + \vb'.
\end{equation*}
Multiplying both sides of the inequality defining $p$ in \eqref{key ineq unit vectors} by $\|\mA\va + \sqrt{m}\vb\|_2+\|\mA\va'+\sqrt{m}\vb'\|_2$ yields
\begin{align*}
	p&=\Pr{\left| \frac{\|\mA\va+\sqrt{m}\vb\|_2^2 - \|\mA\va'+\sqrt{m}\vb'\|_2^2}{\sqrt{\|\va-\va'\|_2^2 + \|\vb-\vb'\|_2^2}} \right| \geq t\big(\|\mA\va + \sqrt{m}\vb\|_2+\|\mA\va' + \sqrt{m}\vb'\|_2\big) }\\
	 &\leq \Pr{\left| \frac{ \ip{\mA(\va-\va')+\sqrt{m}(\vb-\vb')}{\mA(\va+\va')+\sqrt{m}(\vb+\vb')} }{ \sqrt{\|\va-\va'\|_2^2 + \|\vb-\vb'\|_2^2} } \right| \geq t\cdot\|\mA\va+\sqrt{m}\vb\|_2 }\\
     & = \Pr{\left| \ip{\mA\vu+\sqrt{m}\vv}{\mA\vu'+\sqrt{m}\vv'} \right| \geq t\cdot\|\mA\va+\sqrt{m}\vb\|_2}.
\end{align*}
Define the event $\Omega_1$ as
\begin{align*}
 \Omega_1:=\left\{\|\mA\va+\sqrt{m}\vb\|_2\geq \frac{\sqrt{m}}{2}\right\}.
\end{align*}
By the law of total probability, we have
\begin{align*}
	p &\leq\Pr{\underbrace{\left| \ip{\mA\vu+\sqrt{m}\vv}{\mA\vu'+\sqrt{m}\vv'} \right| \geq t \cdot\|\mA\va+\sqrt{m}\vb\|_2}_{:=\Omega_0} }\\
	 &=\Pr{\Omega_0 ~\big|~ \Omega_1}\cdot\Pr{\Omega_1} + \Pr{\Omega_0~\big|~\Omega_1^{c}}\cdot\Pr{\Omega_1^{c}}\\
	 &\leq \Pr{\Omega_0 ~\textrm{and}~ \Omega_1} + \Pr{\Omega_1^{c}}\\
	 &:=p_1 + p_2.
\end{align*}
We can easily bound $p_2$ using Lemma \ref{lemma: concen of cor random mat}
\begin{align*}
	p_2&=\Pr{\Omega_1^{c}}=\Pr{\|\mA\va+ \sqrt{m}\vb\|_2 \leq \frac{\sqrt{m}}{2}}\\
	&=\Pr{\|\mA\va+ \sqrt{m}\vb\|_2 - \sqrt{m} \leq -\frac{\sqrt{m}}{2}}\\
	&\leq \Pr{ \Big| \|\mA\va+ \sqrt{m}\vb\|_2 - \sqrt{m} \Big| \geq \frac{\sqrt{m}}{2} }\\
	&\leq \Pr{ \Big| \|\mA\va+ \sqrt{m}\vb\|_2 - \sqrt{m} \Big| \geq \frac{t}{4} } \quad\text{since}~~ t \leq 2\sqrt{m}\\
	&\leq 2\exp\left\{-\frac{c_1t^2}{16K^4}\right\}\quad\text{since Lemma \ref{lemma: concen of cor random mat}}\\
	&\leq 2\exp\left\{-\frac{c''t^2}{K^4}\right\}.
\end{align*}
To bound $p_1$, note that
\begin{align*}
	p_1 &:= \Pr{\Omega_0 ~\textrm{and}~ \Omega_1} \\
        &\leq \Pr{\Big| \ip{\mA\vu+\sqrt{m}\vv}{\mA\vu'+\sqrt{m}\vv'} \Big| \geq \frac{t\sqrt{m}}{2}}\\
	    &=\Pr{\Big| \ip{\mA\vu}{\mA\vu'} + m\ip{\vv}{\vv'} + \sqrt{m}\ip{\mA\vu}{\vv'} + \sqrt{m}\ip{\mA\vu'}{\vv}\Big| \geq \frac{t\sqrt{m}}{2} }\\
	    &\leq \Pr{\Big| \ip{\mA\vu}{\mA\vu'} + m\ip{\vv}{\vv'} \Big| \geq \frac{t\sqrt{m}}{4} } + \Pr{ \Big| \sqrt{m}\ip{\mA\vu}{\vv'} + \sqrt{m}\ip{\mA\vu'}{\vv}\Big| \geq \frac{t\sqrt{m}}{4} }\\
	    &\leq \Pr{\Big| \ip{\mA\vu}{\mA\vu'} + m\ip{\vv}{\vv'} \Big| \geq \frac{t\sqrt{m}}{4} } + \Pr{ \Big| \sqrt{m}\ip{\mA\vu}{\vv'}\Big| \geq \frac{t\sqrt{m}}{8} } + \Pr{ \Big| \sqrt{m}\ip{\mA\vu'}{\vv}\Big| \geq \frac{t\sqrt{m}}{8} }\\
	&=:p_{1a}+p_{1b}+ p_{1c}.
\end{align*}
We bound the three terms separately. To bound $p_{1a}$, note that
\begin{align*}
	\ip{\vv}{\vv'} &= \frac{ \ip{ \vb - \vb' }{ \vb + \vb'}}{\sqrt{\|\va - \va'\|_2^2 + \|\vb - \vb'\|_2^2 } } \\
	&=\frac{ \|\vb\|_2^2 - \|\vb'\|_2^2}{\sqrt{\|\va - \va'\|_2^2 + \|\vb - \vb'\|_2^2 }}\\
	&=-\frac{ \|\va\|_2^2 - \|\va'\|_2^2}{\sqrt{\|\va - \va'\|_2^2 + \|\vb - \vb'\|_2^2 }}\\
	&=-\ip{\vu}{\vu'}.
\end{align*}
Therefore, we can bound $p_{1a}$ as $p_1$ \eqref{Bound_p1} in Lemma \ref{lemma: concen of cor random mat}
\begin{align*}
	p_{1a} &= \Pr{\Big| \ip{\mA\vu}{\mA\vu'} + m\ip{\vv}{\vv'} \Big| \geq \frac{t\sqrt{m}}{4} }\\
	&=\Pr{\Big| \ip{\mA\vu}{\mA\vu'}-m\ip{\vu}{\vu'} \Big| \geq \frac{t\sqrt{m}}{4} }\\
	&=\Pr{\left|\sum_{i=1}^{m}\left[\ip{\mA_i^T}{\vu} \ip{\mA_i^T}{\vu'} - \ip{\vu}{\vu'}\right] \right| \geq \frac{t\sqrt{m}}{4} }\\
	&=\Pr{\left| \sum_{i=1}^{m}Z_i \right| \geq \frac{t\sqrt{m}}{4} }.
\end{align*}
By assumption, it follows from facts \ref{pro: center} and \ref{Product of subgaus is subexp} that $\{Z_i\}_{i=1}^{m}$ are independent, mean-zero, sub-exponential random variables with
\begin{align*}
	\|Z_i\|_{\psi_1} &= \|\ip{\mA_i^T}{\vu} \ip{\mA_i^T}{\vu'} - \ip{\vu}{\vu'}\|_{\psi_1}\\
	                 & \leq C'\|\ip{\mA_i^T}{\vu} \ip{\mA_i^T}{\vu'}\|_{\psi_1}\\
	                 & \leq C'\|\ip{\mA_i^T}{\vu}\|_{\psi_2}\cdot\|\ip{\mA_i^T}{\vu'}\|_{\psi_2}\\
	                 &\leq C'K^2\|\vu\|_2\cdot\|\vu'\|_2\leq 2C'K^2 = CK^2.
\end{align*}
Using Bernstein's inequality (Fact \ref{Bernstein ineq}) yields
\begin{align*}
	p_{1a} &\leq 2\exp\left\{-c\min\left(\frac{t^2}{16C^2K^4}, \frac{\sqrt{m}t}{4CK^2} \right) \right\}\\
	&\leq 2\exp\left\{-\frac{c}{16C^2K^4}\min\left(t^2, 4CK^2t\sqrt{m} \right)\right\}\\
	&\leq 2\exp\left\{-\frac{c}{16C^2K^4}\min\left(t^2, 2CK^2t^2 \right)\right\} ~~~\text{since }t\leq2\sqrt{m}\\
	&\leq 2\exp\left\{-\frac{ct^2}{16C^2K^4}\right\}~~~\text{choose}~ C ~\text{such that}~ 2CK^2\geq1\\
	&\leq 2\exp\left\{-\frac{c_1t^2}{K^4}\right\}.
\end{align*}
Similar to $p_2$ \eqref{Bound_p2} in Lemma \ref{lemma: concen of cor random mat}, we can easily obtain $p_{1b} \leq 2\exp\left(-{c_2t^2}/{K^4} \right)$ and $p_{1c} \leq 2\exp\left(-{c_3t^2}/{K^4}\right)$.
Combining $p_{1a},p_{1b},$ and $p_{1c}$ yields
	\begin{align*}
	p_1 &\leq 2\exp\left\{-\frac{c_1t^2}{K^4}\right\} + 2\exp\left\{-\frac{c_2t^2}{K^4} \right\} + 2\exp\left\{-\frac{c_3t^2}{K^4}\right\}\\
	&\leq 6\exp\left\{-\frac{c't^2}{K^4}\right\}.
	\end{align*}
Therefore, we have
	\begin{align*}
	p &\leq p_1 + p_2\leq C\exp\left\{-\frac{ct^2}{K^4}\right\}.
	\end{align*}
This completes the proof.
\end{proof}

\noindent\textbf{Step 3:} $(\va,\vb),(\va',\vb')\in\R^n\times\R^m$.

Finally, we prove the increment inequality \eqref{subGaussian_increments} in full generality. Without lost of generality, we assume that $\|\va\|_2^2+\|\vb\|_2^2 = 1$ and $\|\va'\|_2^2+\|\vb'\|_2^2 \geq 1$. Define the unit vector $(\bar{\va'}^T,\bar{\vb'}^T)^T$ with
\begin{align*}
\bar{\va'}:=\frac{\va'}{\sqrt{\|\va'\|_2^2 + \|\vb'\|_2^2}}~~\text{ and }~~\bar{\vb'}:=\frac{\vb'}{\sqrt{\|\va'\|_2^2 + \|\vb'\|_2^2}}.
\end{align*}
Then we have
\begin{align*}
\|X_{\va,\vb}-X_{\va',\vb'}\|_{\psi_2} &\leq\| X_{\va,\vb} - X_{\bar{\va'},\bar{\vb'}} \|_{\psi_2} + \|X_{\bar{\va'},\bar{\vb'}}-X_{\va',\vb'}\|_{\psi_2}\\
 &=:R_1 + R_2.
\end{align*}
By Lemma \ref{lemma: double units}, $R_1 \leq CK^2\cdot\sqrt{\left\|\va-\bar{\va'}\right\|_2^2 + \left\|\vb-\bar{\vb'}\right\|_2^2}$. Since $(\bar{\va'},\bar{\vb'})$ and $(\va',\vb')$ are colinear, we have
\begin{align*}
R_2 &= \|X_{\bar{\va'},\bar{\vb'}}-X_{\va',\vb'}\|_{\psi_2}\\
&=\sqrt{\|\va'-\bar{\va'}\|_2^2 + \|\vb'-\bar{\vb'}\|_2^2}\cdot\|X_{\bar{\va'},\bar{\vb'}}\|_{\psi_2}\\
&\leq CK^2\cdot\sqrt{\|\va'-\bar{\va'}\|_2^2 + \|\vb'-\bar{\vb'}\|_2^2}.
\end{align*}
The last inequality follows from Lemma \ref{lemma: concen of cor random mat}.

Combining $R_1$ and $R_2$ yields
\begin{align*}
R_1+R_2\leq \sqrt{2}CK^2\cdot\sqrt{\|\va-\va'\|_2^2 + \|\vb-\vb'\|_2^2 },
\end{align*}
where we have used the fact \cite[Exercise 9.1.7]{vershynin2016book} that
\begin{align*}
	\sqrt{\|\va-\bar{\va'}\|_2^2 + \|\vb-\bar{\vb'}\|_2^2 } + \sqrt{\|\bar{\va'}-\va'\|_2^2 + \|\bar{\vb'}-\vb'\|_2^2 } \leq \sqrt{2}\sqrt{\|\va-\va'\|_2^2 + \|\vb-\vb'\|_2^2 }.
\end{align*}
This completes the proof of Lemma \ref{them: sub-Gaussian process}.

\section*{Acknowledgment}
The authors thank Huan Zhang and Xu Zhang for helpful discussions.

\bibliographystyle{IEEEtran}
\bibliography{IEEEabrv,myref}


\end{document}